\documentclass[11pt]{extarticle}
\usepackage{amsmath,amssymb,amsthm}
\usepackage{graphicx}
\usepackage[margin=1in]{geometry}
\usepackage{bbm}
\usepackage{graphicx}
\usepackage{graphicx}
\usepackage{mleftright}
\usepackage{pgf}
\usepackage{tikz}
\usepackage{verbatim}
\usepackage{braket}
\usetikzlibrary{decorations.markings}

\usepackage{enumerate}
\usepackage{hyperref}
\usepackage[capitalise]{cleveref}
\usepackage{pgfplots}
\usetikzlibrary{arrows}
\newtheorem{theorem}{Theorem}[section]

\newtheorem{corollary}[theorem]{Corollary}
\newtheorem{claim}[theorem]{Claim}

\newtheorem{definition}[theorem]{Definition}
\newtheorem{lemma}[theorem]{Lemma}
\newcommand{\E}{\mathbb{E}}
\newcommand{\p}{\mathbb{P}}

\title{Lower Bounds for XOR of Forrelations}
\date{}
\author{Uma Girish\thanks{Department of Computer Science, Princeton University. Research supported by the Simons Collaboration on Algorithms and Geometry, by a Simons Investigator Award and by the National Science Foundation grant No. CCF-1714779.} 
\and Ran Raz\thanks{Department of Computer Science, Princeton University. Research supported by the Simons Collaboration on Algorithms and Geometry, by a Simons Investigator Award and by the National Science Foundation grant No. CCF-1714779.} 
\and Wei Zhan\thanks{Department of Computer Science, Princeton University. Research supported by the Simons Collaboration on Algorithms and Geometry, by a Simons Investigator Award and by the National Science Foundation grant No. CCF-1714779.}}

\begin{document}
\maketitle

\begin{abstract}
The Forrelation problem, first introduced by Aaronson~\cite{aaronson10} and Aaronson and Ambainis~\cite{aaronsonambainis}, is  a well studied computational problem in the context of separating quantum and classical computational models. Variants of this problem were used to give tight separations between quantum and classical query complexity~\cite{aaronsonambainis}; the first separation between poly-logarithmic quantum query complexity and bounded-depth circuits of super-polynomial size, a result that also implied an oracle separation of the classes BQP and PH~\cite{raztal}; and improved separations between quantum and classical communication complexity~\cite{grt}. In all these separations, the lower bound for the classical model only holds when the advantage of the protocol (over a random guess) is more than $\approx 1/\sqrt{N}$, that is, the success probability is larger than $\approx 1/2 + 1/\sqrt{N}$. This is unavoidable as $\approx 1/\sqrt{N}$ is the correlation between two coordinates of an input that is sampled from the Forrelation distribution, and hence there are simple classical protocols that achieve advantage $\approx 1/\sqrt{N}$, in all these models.

To achieve separations when the classical protocol has smaller advantage, we study in this work the \textsc{xor} of $k$ independent copies of (a variant of) the Forrelation function (where $k\ll N$).
We prove a very general result that shows that any family of Boolean functions that is closed under restrictions, whose Fourier mass at level $2k$ is bounded by $\alpha^k$ (that is, the sum of the absolute values of all Fourier coefficients at level $2k$ is bounded by $\alpha^k$), cannot compute the \textsc{xor} of $k$ independent copies of the Forrelation function with advantage better than $O\left(\frac{\alpha^k}{{N^{k/2}}}\right)$.
This is a strengthening of a result of \cite{chlt}, that gave a similar statement for $k=1$, using the technique of~\cite{raztal}.
We give several applications of our result. In particular, we obtain the following separations:

{\bf Quantum versus Classical Communication Complexity:} We give the first example of a partial Boolean function that can
be computed by a simultaneous-message quantum protocol with communication
complexity $\mbox{polylog}(N)$ (where Alice and Bob also share $\mbox{polylog}(N)$ EPR pairs), and such that, any classical randomized protocol of communication complexity at most $\tilde{o}(N^{1/4})$, with any number of rounds,  has quasipolynomially small advantage over a random guess. Previously, only separations where the classical protocol has polynomially small advantage were known between these models~\cite{gavinsky,grt}.

{\bf Quantum Query Complexity versus Bounded Depth Circuits:}
We give the first example of a partial Boolean function
that has a quantum query algorithm with query complexity $\mbox{polylog}(N)$, and such that,
any constant-depth circuit of quasipolynomial size has quasipolynomially small advantage over a random guess. Previously, only separations where the constant-depth circuit has polynomially small advantage were known~\cite{raztal}.

\end{abstract}

\section{Introduction}

Several recent works used Fourier analysis to prove lower bounds for computing (variants of) the Forrelation (partial) function of~\cite{aaronson10,aaronsonambainis}, in various models of computation and communication~\cite{raztal,chlt,grt}. These works show that for many computational models, when analyzing the success probability of computing the Forrelation function, it's
sufficient to bound the contribution of Fourier coefficients at level~2, ignoring all other Fourier coefficients~\cite{raztal,chlt}. This holds for any computational model that is closed under restrictions and
is proved by analyzing the Forrelation distribution as a distribution resulting from a certain random walk, rather than analyzing it directly.

While this is a powerful technique, it could only be used to bound computations of the Forrelation function with advantage (over a random guess) larger than 
$\approx 1/\sqrt{N}$, that is, computations with success probability larger than $\approx 1/2 + 1/\sqrt{N}$.
Roughly speaking, this is because the bound on the Fourier coefficients at level~2 of the Forrelation function is $\approx O\big(1/\sqrt{N}\big)$.

In this work, we study the \textsc{xor} of $k$ independent copies of the Forrelation function of~\cite{raztal} (where $k< o(N^{1/50})$).
We show that for many computational models, when analyzing the success probability of computing the \textsc{xor} of $k$ independent copies of the Forrelation function, it's
sufficient to bound the contribution of Fourier coefficients at level~$2k$, ignoring all other Fourier coefficients. Our proof builds on the techniques of~\cite{raztal}, and followup works~\cite{chlt,grt}, by analyzing a ``product'' of $k$ random walks, one for each of the independent copies of the Forrelation function. This can be viewed as a random walk with a $k$-dimensional time variable.

Consequently, we obtain a very general lower bound that
shows that any family of Boolean functions that is closed under restrictions, whose Fourier mass at level $2k$ is bounded by $\alpha^k$ (that is, for every function in the family, the sum of the absolute values of all Fourier coefficients at level $2k$ is bounded by $\alpha^k$), cannot compute the \textsc{xor} of $k$ independent copies of the Forrelation function with advantage better than $O\left(\frac{\alpha^k}{{N^{k/2}}}\right)$, that is, with success probability larger than $\frac{1}{2}+O\left(\frac{\alpha^k}{{N^{k/2}}}\right)$.
This is a strengthening of a result of \cite{chlt}, that gave a similar statement for $k=1$, using the technique of~\cite{raztal}.

We note that the requirement that the family of Boolean functions is closed under restrictions is satisfied by essentially all non-uniform computational models. The requirement of having a good bound on the Fourier mass at level $2k$ is satisfied by several central and well-studied computational models (see for example~\cite{chhl} for a recent discussion). In particular, we focus in this work on three such models: communication complexity, query complexity (decision trees) and bounded-depth circuits.
We note that our result is valid for any $k< N^{c}$, for some constant $c>0$, and hence it can be used to prove lower bounds for circuits/protocols with exponentially small advantage, in all these models. However, for the applications of separating quantum and classical computational models, we take $k$ to be poly-logarithmic in $N$, so that we have quantum protocols of poly-logarithmic cost.
We use our main theorem to give several separations between quantum and classical computational models.

\subsection{Communication Complexity}

Quantum versus classical separations in communication complexity have been studied for more than two decades in numerous works. We briefly summarize the history of quantum advantage in communication complexity of partial functions, that is most relevant for us: First, Buhrman, Cleve and Wigderson proved an exponential separation between zero-error simultaneous-message quantum communication complexity (without entanglement) and classical deterministic communication complexity~\cite{buhrman}. For the bounded-error model, Raz showed an exponential separation between two-way quantum communication complexity and two-way randomized communication complexity~\cite{raz}. Gavinsky et al (building on Bar-Yossef et al~\cite{Bar-YossefJK04}) gave an exponential separation between one-way quantum communication complexity and one-way randomized communication complexity~\cite{gavinskyetal}. Klartag and Regev gave an exponential separation between one-way quantum communication complexity and two-way randomized communication complexity~\cite{klartagregev}.
The state of the art separation, by Gavinsky, gave an exponential separation between simultaneous-message quantum communication complexity (with entanglement)
and two-way randomized communication complexity~\cite{gavinsky}.
An alternative proof for Gavinsky's result was recently given by~\cite{grt}, as a followup to~\cite{raztal,chlt}, and had the additional desired property that in the quantum protocol, the time complexity of all the players is poly-logarithmic.

\subsubsection*{Our Result:}

In all these works, the lower bounds for classical communication complexity only hold when the advantage of the protocol (over a random guess) is more than $\approx 1/\sqrt{N}$, that is, the success probability is larger than $\approx 1/2 + 1/\sqrt{N}$.

In this work, we give a partial Boolean function that can be computed by a simultaneous-message quantum protocol with communication complexity $\mbox{polylog}(N)$ (where Alice and Bob also share $\mbox{polylog}(N)$ EPR pairs), and such that, any classical randomized protocol of communication complexity at most $\tilde{o}(N^{1/4})$, with any number of rounds, has quasipolynomially small advantage over a random guess. This qualitatively matches the results of~\cite{gavinsky,grt} and has the additional desired property that the lower bound for the classical communication protocol holds for quasipolynomially small advantage, rather than polynomially small advantage.
Moreover, as in~\cite{grt}, the quantum protocol in our upper bound has the additional property of being {\it efficiently implementable}, in the sense that it can be described by quantum circuits of size $\mbox{polylog}(N)$, with oracle access to the inputs.

To prove this result we use  the \textsc{xor} of $k$ independent copies of the Forrelation function,
lifted to communication complexity using \textsc{xor} as the gadget~\cite{razxor}, as in~\cite{grt}. The quantum upper bound is simple. For the classical lower bound, we use  ideas from \cite{grt} to bound the level-$2k$ Fourier mass. This, along with our main theorem implies the desired separation. Our bounds for the level-$2k$ Fourier mass may be interesting in their own right and are proved in Section~7.

\subsubsection*{Related Work:}

We note that an exponential separation between {\bf two-way} quantum communication complexity
and two-way randomized communication complexity, with quasipolynomially small advantage, can be proved by a combination of several previous results, as follows:

Start with an existing separation between quantum and classical query complexity, such as the one of~\cite{aaronsonambainis}. Use Drucker's \textsc{xor}~lemma for randomized decision tree~\cite{drucker} to get a separation between quantum and classical query complexity, where the classical protocol has quasipolynomially small advantage. Finally, use the recent lifting theorem of~\cite{bppip} to lift the result to communication complexity.
To the best of our knowledge, this separation was not previously observed.

It follows from these works that there exists a function computable in the quantum two-way model in communication complexity $\mbox{polylog}(N)$, for which randomized protocols of cost $\tilde{o}(\sqrt{N})$ have at most quasipolynomially small advantage. While the lower bound is for cost $\tilde{o}(\sqrt{N})$ protocols, which is quantitatively stronger than our lower bound for cost $\tilde{o}(N^{1/4})$ protocols, the quantum upper bound in this result seems to require two rounds of communication, while our function is computable in the simultaneous model when Alice and Bob share entanglement.

\subsection{Bounded Depth Circuits}

Separations of quantum query complexity and bounded-depth classical circuit complexity have been studied in the context of oracle separations of the classes BQP and PH. 
An example of a partial Boolean function (Forrelation)
that has a quantum query algorithm with query complexity $\mbox{polylog}(N)$, and such that,
any constant-depth circuit of quasipolynomial size has polynomially small advantage over a random guess, was given in~\cite{raztal}. This result implied an oracle separation of the classes BQP and PH.

Here, we give the first example of a partial Boolean function (\textsc{xor} of $k$ copies of Forrelation)
that has a quantum query algorithm with query complexity $\mbox{polylog}(N)$, and such that,
any constant-depth circuit of quasipolynomial size has {\bf quasipolynomially} small advantage over a random guess. 

For the proof,  we use our main theorem, together with Tal's bounds on the level-$2k$ Fourier mass of bounded-depth circuits~\cite{talac}.

\subsection{Decision Trees}

The query complexity model (also known as black box model or decision-tree complexity) has played a central role in the study of quantum computational complexity.
Quantum advantages in query complexity (decision trees) have been demonstrated for partial functions in various settings and numerous works. For example, Aaronson and Ambainis~\cite{aaronsonambainis} showed that the Forrelation problem can be solved by one quantum query, while its randomized query complexity is  $\Omega(\sqrt{N}/\log N)$.

For classical randomized query complexity, there is a known \textsc{xor} lemma, proved by Drucker~\cite{drucker}. In particular, Theorem 1.3 of~\cite{drucker}, along with the result of \cite{aaronsonambainis} gives
a partial function (\textsc{xor} of $\mbox{polylog}(N)$ copies of Forrelation) that can be computed  by a quantum query algorithm with $\mbox{polylog}(N)$ queries, while  every classical randomized algorithm that makes $\tilde{o}(N^{1/2})$ queries, has quasipolynomially small advantage.

Our main theorem implies a different proof for this result, 
using Tal's recent bounds on the level-$2k$ Fourier mass of decision trees~\cite{tal}.

\subsection{The Main Theorem}

Our functions are obtained by taking an \textsc{xor} of several copies of a variant of the Forrelation problem, as defined in \cite{raztal}.

Let $N=2^n$ for sufficiently large $n\in \mathbb{N}$. Let $k\in \mathbb{N}$ be a parameter. We assume that $k=o(N^{1/50})$.
Let $\bf{\epsilon=\frac{1}{60k^2 \ln N}}$ be a parameter.

Let $H_N$ denote the $N\times N$ normalized Hadamard matrix whose entries are either $-\frac{1}{\sqrt{N}}$ or $\frac{1}{\sqrt{N}}$.
Let $$forr(z) := \frac{1}{N} \left< z_2 , H_N z_1 \right> $$ denote the {\it Forrelation} of a vector  $z=(z_1,z_2)$, where $z_1,z_2 \in \mathbb{R}^{N}$.
The {\bf Forrelation Decision Problem} is the partial Boolean function $F:\{-1,1\}^{2N}\rightarrow \{-1,1\}$ defined at $z\in \{-1,1\}^{2N}$ by
\[F(z):=\begin{cases} -1 & \text{ if } forr(z ) \ge \epsilon/2 \\ 1 & \text{ if } forr(z) \le \epsilon/4\\  \text{undefined} & \text{ otherwise } \end{cases} \]
The {\bf $\oplus^k$ Forrelation Decision Problem} $F^{(k)}:\{-1,1\}^{2kN}\rightarrow \{-1,1\}$ is defined as the \textsc{xor} of $k$ independent copies of $F$. More precisely, for every $z_1,\ldots,z_k\in \{-1,1\}^{2N}$, let $$F^{(k)}(z_1,\ldots,z_k):=\prod_{j=1}^k F(z_j).$$

For our separation results, we take the function $F^{(k)}$, where $k=\lceil\log^2 N\rceil$. For our communication complexity separation we take the lift of $F^{(k)}$ with \textsc{xor} as the gadget. The quantum upper bounds in all these separation results are quite simple. Moreover, all the quantum algorithms in our upper bounds have the additional advantage of being {\it efficiently implementable}, in the sense that they can be described by quantum circuits of size $\mbox{polylog}(N)$, with oracle access to the inputs. 

Our main contribution is the classical lower bound. Towards this, our main theorem provides an upper bound on the maximum correlation of $F^{(k)}$ with any family of Boolean functions, in terms of the maximum level-$2k$ Fourier mass of a function in the family.

{\it {\bf Main Theorem} {\it (Informal)} There exist two distributions, $\sigma_0^{(k)}$ and $\sigma_1^{(k)}$, on the \textsc{no} and \textsc{yes} instances of $F^{(k)}$, respectively, with the following property. Let $\mathcal{H}$ be a family of Boolean functions, each of which maps $\{-1,1\}^{2kN}$ into $[-1,1]$. Assume that $\mathcal{H}$ is closed under restrictions. For $H\in \mathcal{H}$, let $L_{2k}(H):=\sum_{|S|=2k}|\widehat{H}(S)|$. Let $\alpha\in \mathbb{R}$ be such that $\alpha^k:=\underset{H\in \mathcal{H}}{\sup} \left( L_{2k}(H),1\right)$. Then, for every $H\in \mathcal{H}$,
\[ \left| \underset{ z\sim \sigma_0^{(k)}}{\E}\left[ H(z) \right]  -  \underset{ z\sim \sigma_1^{(k)}}{\E}\left[ H(z) \right] \right| \le O\left( \frac{\alpha^k}{N^{k/2}}\right)  \]}

Our main theorem implies that functions in $\mathcal{H}$ cannot correlate with $F^{(k)}$ by more than $\frac{1}{2}+O\left(\frac{\alpha^k}{N^{k/2}} \right)$.
For the applications, we instantiate $\mathcal{H}$ with the class of functions computed by classical protocols of small cost.

\subsection{Overview of Proof of the Main Theorem for $k=2$}

Our proof builds on the techniques of~\cite{raztal}, and followup works~\cite{chlt,grt}, which, in turn, used a key idea from~\cite{chhl}.
We will now give an overview of the proof of the Main Theorem for the special case $k=2$, where one can already see most of the key ideas.

We start by recalling the hard distributions for $k=1$, as in \cite{raztal}. The {\bf distribution $\mathcal{U}$ on \textsc{no} instances} of $F$ is the uniform distribution $U_{2N}$ on $\{-1,1\}^{2N}$. It can be shown that a bit string drawn uniformly at random almost always has low Forrelation. The {\bf distribution $\mathcal{G}$ on \textsc{yes} instances} of $F$ is the Gaussian distribution with mean 0 and covariance matrix $\epsilon\begin{bmatrix} \mathbb{I}_N & H_N \\ H_N & \mathbb{I}_N \end{bmatrix}$. It can be shown that a vector drawn from this distribution almost always has high Forrelation (at least $\epsilon/2$). Although $\mathcal{G}$ is not a distribution over $\{-1,1\}^{2N}$, this can be fixed (by probabilistically rounding the values) and we ignore this issue in the proof overview.

Our hard distributions for $k\ge2$ are obtained by naturally lifting these distributions. The {\bf distribution $\mu_0$ on \textsc{no} instances} of $F^{(2)}$ is $\frac{1}{2}\left( \mathcal{U}\times \mathcal{U} + \mathcal{G}  \times\mathcal{G} \right)$. The {\bf distribution $\mu_1$ on \textsc{yes} instances} is $\frac{1}{2}\left( \mathcal{U} \times\mathcal{G}+ \mathcal{G} \times\mathcal{U}\right)$. It can be shown that these distributions indeed have almost all their mass on the \textsc{yes} and \textsc{no} instances of $F^{(2)}$, respectively.

Throughout this proof, we identify functions in $\mathcal{H}$ with their unique multilinear extensions.
Using this identification, it follows that for all $H\in \mathcal{H}$ and $z_0\in \mathbb{R}^{4N}$, we have $\E_{z\sim \mathcal{U}}[H(z_0+(z,0))]=\E_{z\sim \mathcal{U}}[H(z_0+(0,z))]=\E_{z\sim \mathcal{U}^2}[H(z_0+z)]=H(z_0)$.

\subsubsection*{Bounding the Advantage of $H$ in Distinguishing $p\cdot \mu_0$ and $p\cdot \mu_1$, for Small $p$:}

As in~\cite{raztal,chlt}, in
order to show that functions in $\mathcal{H}$ can't distinguish between $\mu_0$ and $\mu_1$, we first show that they can't distinguish between $p\cdot \mu_0$ and $p\cdot \mu_1$, for small $p$. We show that for every $H\in \mathcal{H}$, and $p\le\frac{1}{2N}$,
\begin{align*}\begin{split}
\left| \underset{z\sim p\cdot \mu_0}{\E}[H(z)]- \underset{z\sim p\cdot \mu_1}{\E}[H(z)]\right|  &\triangleq \frac{1}{2}  \left| \underset{\substack{z_1\sim p\cdot \mathcal{G} \\ z_2\sim  p\cdot\mathcal{G}}}{\E} \left [H(z_1,z_2) -H(z_1,0)-H(0,z_2) + H(0,0)\right] \right| \\
&\le p^4\cdot O\left( \frac{L_4(H)}{N} \right) + O(p^6 N^{1.5})
\end{split}\end{align*}
This claim is analogous to Claim~20 from~\cite{chlt}. For sufficiently small $p$, the second term in the R.H.S. of the inequality is negligible, compared to the first term. To prove this inequality, we use the Fourier expansion of $H$ in the L.H.S. and bound the difference between the moments of $p\cdot\mu_0$ and $p\cdot\mu_1$. We show that $p\cdot\mu_0$ and $p\cdot\mu_1$ agree on moments of degree less than~4, so these moments don't contribute to the difference. We then show that the contribution of the moments of degree~4 is $L_4(H)\cdot O\left( \frac{p^4}{N} \right)$ and the contribution of  moments of higher degrees is $O(p^6 N^{1.5})$.

\subsubsection*{Bounding the Advantage of $H(z_0 + z)$ in Distinguishing $p\cdot \mu_0$ and $p\cdot \mu_1$, for Small $p$:}

Next, as in~\cite{raztal,chlt},
we show a similar statement for the function $H(z_0+z)$ of $z$, where $z_0$ is not too large. We show that
for every $H\in \mathcal{H}$, and every $z_0\in [-1/2,1/2]^{2kN}$ and $p\le\frac{1}{2N}$,
\begin{align}\label{inequality:shifted}\begin{split}
& \frac{1}{2}  \left| \underset{\substack{z_1\sim p\cdot \mathcal{G} \\ z_2\sim  p\cdot\mathcal{G}}}{\E} \left [H(z_0+(z_1,z_2)) -H(z_0+(z_1,0))-H(z_0+(0,z_2)) + H(z_0)\right] \right| \\
&\le p^4\cdot O\left( \frac{L_4(H)}{N} \right) + O(p^6 N^{1.5})
\end{split}\end{align}
The proof of this inequality is similar to the proof of Claim 19 of~\cite{chlt}, using key ideas from~\cite{chhl}, and relies on the multilinearity of functions in $\mathcal{H}$ and the closure of $\mathcal{H}$ under restrictions.

\subsubsection*{A Random Walk with Two-Dimensional Time Variable:}

This is the main place where our proof differs from the one of~\cite{raztal} and followup works~\cite{chlt,grt}. In all these works the Forrelation distribution was ultimately analyzed as the distribution obtained by a certain random walk. Here, we consider a product of two random walks, which can also be viewed as a random walk with two-dimensional time variable.

Let $T= 16N^4 $ and $p=\frac{1}{\sqrt{T}}$. Let $z_1^{(1)},z_2^{(1)},\ldots,z_1^{(T)},z_2^{(T)}\sim p\cdot \mathcal{G}$ be independent samples. Let $t=(t_1,t_2)$ for $t_1,t_2\in \{0,\ldots,T\}$. Let $z^{\le(t)}:=\left( \sum_{i=1}^{t_1} z_1^{(i)}, \sum_{i=1}^{t_2} z_2^{(i)}\right)$. Note that $z^{\le(t)}$ is distributed according to $ (p\sqrt{t_1}\cdot \mathcal{G})\times(p\sqrt{ t_2}\cdot \mathcal{G})$. In particular, $z^{\le(T,T)}$ is distributed according to $\mathcal{G}\times \mathcal{G}$. This implies that
\[(*):=\underset{z\sim  \mu_0}{\E}[H(z)]- \underset{z\sim  \mu_1}{\E}[H(z)]  \triangleq  \frac{1}{2}  {\E} \left [H(z^{\le(T,T)}) -H(z^{\le(T,0)})-H(z^{\le(0,T)}) + H(0,0)\right]\]
We now rewrite $(*)$ as follows.
\begin{equation} \label{introtelescopic} (*)= \frac{1}{2}  \underset{\substack{t_1\in [T]\\t_2\in [T]}}{\sum} \E \left[H(z^{\le(t_1,t_2)}) -H(z^{\le(t_1-1,t_2)}) - H(z^{\le(t_1,t_2-1)})+H(z^{\le(t_1-1,t_2-1)})\right]
\end{equation}
The last equation follows by a two-dimensional telescopic cancellation, as depicted in Figure~1.
This turns out to be a powerful observation.
Note that for every fixed $t = (t_1,t_2)$, the random variable $z^{\le(t)}-z^{\le(t-(1,1))}\triangleq (z_1^{(t_1)},z_2^{(t_2)})$ is distributed according to $ p\cdot \mathcal{G}^2$, by construction. We can thus apply Inequality\eqref{inequality:shifted}, setting $z_0=z^{\le(t-(1,1))}$. This, along with the Triangle-Inequality implies that
\begin{align*} \begin{split}
|(*)| &\le\frac{1}{2}  \underset{\substack{t_1\in [T]\\t_2\in [T]}}{\sum} \left| \E \left[H(z^{\le(t_1,t_2)}) -H(z^{\le(t_1-1,t_2)}) - H(z^{\le(t_1,t_2-1)})+H(z^{\le(t_1-1,t_2-1)})\right]  \right|  \\
&\le \frac{1}{2}  \underset{\substack{t_1\in [T]\\t_2\in [T]}}{\sum}\left( p^4\cdot O\left( \frac{L_4(H)}{N} \right) + O\left(p^6N^{1.5} \right) \right)  \quad \;\;\;\;\;\; \text{ by Inequality~\eqref{inequality:shifted}} \\
&=  O\left( \frac{L_4(H)}{N} \right) + o\left(\frac{1}{N} \right)  \quad \;\;\;\;\;\; \text{ since }T= 16N^4 =\frac{1}{p^2}
\end{split}\end{align*}

This completes the proof overview for $k=2$, albeit with many details left out.

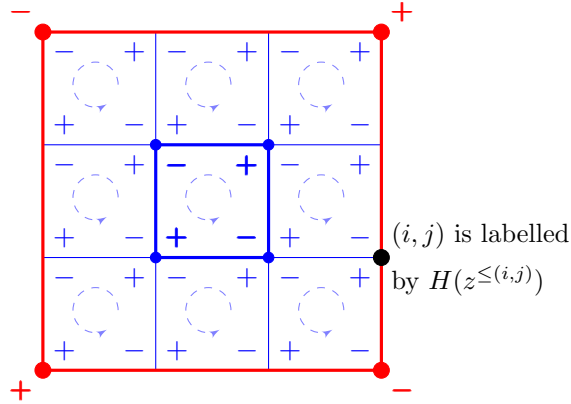
\begin{figure}[h]
\centering
\begin{tikzpicture}

  \draw[blue,step=1.5] (0,0) grid (4.5,4.5);
    \filldraw[red] (4.5,0) circle (3pt);
       \node[anchor=north west,red] at (4.5,0) {$\pmb{-}$};

     \filldraw[red] (0,4.5)  circle (3pt);
            \node[anchor=south east,red] at (0,4.5) {$\pmb{-}$};

       \filldraw[red] (4.5,4.5)  circle (3pt);
                   \node[anchor=south west,red] at (4.5,4.5) {$\pmb{+}$};

       \filldraw[red] (0,0)  circle (3pt);
                   \node[anchor=north east,red] at (0,0) {$\pmb{+}$};

    \draw[ blue, very thick] (1.5,1.5)--(3,1.5);
    \draw[ blue, very thick] (3,1.5)--(3,3);
    \draw[ blue, very thick] (3,3)--(1.5,3);
    \draw[blue, very thick] (1.5,3) --(1.5,1.5);

    \draw[very thick,red] (0,0)--(4.5,0);
    \draw[ very thick,red] (4.5,0)--(4.5,4.5);
    \draw[ very thick,red] (4.5,4.5)--(0,4.5);
    \draw[ very thick,red] (0,4.5) --(0,0);

  \filldraw[black] (4.5,1.5)  circle (3pt);
    \node[right] at (4.5,1.8) {\small $(i,j)$ is labelled};
  \node[right] at (4.5,1.2) {\small by $H(z^{\le(i,j)})$};

  \filldraw[blue] (3,3)  circle (2pt);
  \filldraw[blue] (3,1.5)  circle (2pt);
  \filldraw[blue] (1.5,1.5)  circle (2pt);
  \filldraw[blue] (1.5,3)  circle (2pt);
        \node[blue, anchor=north east] at (3,3) {$\pmb{+}$};
	\node[blue, anchor=south west] at (1.5,1.5) {$\pmb{+}$};
	\node[blue, anchor=north west] at (1.5,3) {$\pmb{-}$};
	\node[blue, anchor=south east] at (3,1.5) {$\pmb{-}$};

 \foreach \i in {0,1,2}
      \foreach \j in {0,...,2}
      {
        \node[blue, anchor=south west] at (1.5*\i,1.5*\j) {${+}$};
	\node[blue, anchor=north east] at (1.5*\i+1.5,1.5*\j+1.5) {${+}$};
	\node[blue, anchor=south east] at (1.5*\i+1.5,1.5*\j) {${-}$};
	\node[blue, anchor=north west] at (1.5*\i,1.5*\j+1.5) {${-}$};

        \draw[->,>=latex',blue!50,dashed] (1.5*\i+1,1.5*\j+0.8) arc[radius=0.3,start angle=0,delta angle=300];
      }

\end{tikzpicture}
\caption{Consider the $(T+1)\times (T+1)$ grid whose vertices are indexed by $v\in (\{0\}\cup[T])^2$. Each vertex $v$ is labelled by $H(z^{\le(v)})$. Each rectangle has a sign on its vertices as defined in Figure 1 and the label of a rectangle is the sum of signed labels of its vertices. The sum of labels of all $1\times 1$ rectangles equals the label of the larger $T\times T$ rectangle. This is exactly the content of \cref{introtelescopic}. }
\end{figure}

\subsection{Organization of the Paper}

We present the preliminaries regarding Forrelation in Section~2 and state our main theorems in Section~3. In Section~4, we show how to bound the advantage of $H$ in distinguishing between $p\cdot \mu_0$ and $p\cdot \mu_1$, for Small $p$. In Section~5, we show how to bound the advantage of $H(z_0 + z)$ in distinguishing between $p\cdot \mu_0$ and $p\cdot \mu_1$, for Small $p$. In Section 6, we give the analysis of our random walk with $k$-dimensional time variable. Section 7 contains the proofs of the quantum-classical separations.

\section{Preliminaries}

For $n\in \mathbb{N}$, we use $[n]$ to denote the set $\{1,2,\ldots,n\}$. We typically use $N$ to refer to $2^n$. For a set $S\subseteq [n]$, let $\bar{S}:=[n]\setminus S$ denote the complement of $S$. For sets $S\subseteq[n],T\subseteq[m]$, we typically use $S\times T:= \{ (s,t):s\in S,t\in T\}$ denote the set product of $S$ and $T$. Sometimes, we use the notation $(S,T)$. Note that the map $(i,j)\rightarrow m(i-1) + j$ is a bijection between $[n]\times [m]$ and $[nm]$. Using this identification, $S\times T $ is a subset of $[nm]$.
We identify subsets $S\subseteq[n]$ with their $\{0,1\}$ indicator vector, that is, the vector $S\in\{0,1\}^n$ such that for each $j\in [n]$, $S_j=1$ if and only if $j\in S$.

Let $v\in \mathbb{R}^n$. For $i\in [n]$, we refer to the $i$-th coordinate of $v$ by $v_i$ or $v(i)$. For $x,y\in \mathbb{R}^n$, let $x\cdot y\in \mathbb{R}^n$ be the pointwise product between $x$ and $y$. This is the vector whose $i$-th coordinate is $x_iy_i$, for every $i\in [n]$. Let $\left<x , y\right>$ denote the real inner product between $x$ and $y$. For $x,y\in \{0,1\}^n$, let $\left< x, y\right>_2:=\sum_{i=1}^n x_i y_i \mod 2$ denote the mod 2 inner product between $x$ and $y$. We use $\mathbb{I}_n$ to denote the $n\times n$ identity matrix. We use $0$ to denote the zero vector in arbitrary dimensions.

\paragraph*{Distributions} For a probability distribution $D$, let $x\sim D$ denote a random variable $x$ sampled according to $D$. For distributions $D_1$ and $D_2$, we use $D_1\times D_2$ to denote the product distribution defined by sampling $(x,y)$ where $x\sim D_1$ and $y\sim D_2$ are sampled independently. For $n\in \mathbb{N}$ and a distribution $D$, let $D^n$ denote the product of $n$ distributions, each of which is $D$. Let $\mu\in \mathbb{R}^n$ be a vector and $\Sigma\in \mathbb{R}^{n\times n}$ be a positive semi-definite matrix. We use $\mathcal{N}(\mu,\Sigma)$ to refer to the $n$-dimensional Gaussian distribution with mean $\mu$ and covariance matrix $\Sigma$. Let $U_n$ denote the uniform distribution on $\{-1,1\}^n$. For a distribution $D$ over $\mathbb{R}^n$ and $a\in \mathbb{R}^n$, let $a+D$ refer to the distribution obtained by sampling $z\sim D$ and returning $z+a$. For $P\in \mathbb{R}^n$ and a distribution $D$ over $\mathbb{R}^n$, let $P\cdot D$ denote the distribution obtained by sampling $x\sim D$ and returning $P\cdot x$. For $p\in \mathbb{R}$, we use $p\cdot D$ to denote the distribution obtained by sampling $x\sim D$ and returning $p x$. For $I\subseteq [n]$, let $\widehat{D}(I):= \underset{z\sim D}{\E} \left[ \prod_{i\in I}z_i   \right]$ refer to the $I$-th moment of $D$.

\paragraph*{Concentration Inequalities} We make use of the following concentration inequalities. The first is the Gaussian Concentration Inequality~\cite{gaussian} which states that $ \underset{z\sim \mathcal{N}(0,1)}{\p}[z\ge t] \le e^{-t^2/2}$. We also use the following concentration inequality for the Chi-Squared distribution.~\cite{chi}
\[ \underset{z_1,\ldots,z_n\sim \mathcal{N}(0,1)}{\p}\left[ \left|\frac{1}{n}\sum_{i=1}^n z_i^2 - 1\right| \ge t\right] \le 2e^{-nt^2/8} \quad\quad\text{ for all } t\in (0,1)\]

\paragraph*{Fourier Analysis} We refer to $\{-1,1\}^n$ as the Boolean hypercube in $n$ dimensions. Let $\mathcal{F}:=\{f:\{-1,1\}^n\rightarrow \mathbb{R} \}$ denote the real vector space of all Boolean functions on $n$ variables. There is an inner product on this space as follows. For $f,g\in \mathcal{F}$, let $\left< f,g\right>:=\E_{x\sim U_n}[f(x)g(x)]$. For every $S\subseteq [n]$, there is a character function $\chi_S:\{-1,1\}^n\rightarrow \{-1,1\}$ defined at $x\in \{-1,1\}^n$ by $\chi_S(x):=\prod_{i\in S} x_i$. The set of character functions $\{\chi_S\}_{S\subseteq [n]}$ forms an orthonormal basis for $\mathcal{F}$. For $f\in \mathcal{F}$ and $S\subseteq[n]$, let $\widehat{f}(S):=\left< f,\chi_S\right>$ denote the $S$-th Fourier coefficient of $f$. Note that for all $f\in \mathcal{F}$, we have $f=\sum_{S\subseteq [n]} \widehat{f}(S)\chi_S$. For $f\in \mathcal{F}$, the multilinear extension of $f$ is the unique multilinear polynomial $\tilde{f}:\mathbb{R}^n\rightarrow\mathbb{R}$ which agrees with $f$ on $\{-1,1\}^n$. For every $S\subseteq [n]$, the multilinear extension of $\chi_S$ is the monomial $\prod_{i\in S}x_i$. This implies that the multilinear extension of $f\in \mathcal{F}$ is $\sum_{S\subseteq[n]} \widehat{f}(S)\prod_{i\in S}x_i$. Henceforth, we identify Boolean functions with their multilinear extensions. With this identification, it can be shown that functions in $\mathcal{F}$ which map $\{-1,1\}^{n}$ into $[-1,1]$ also map $[-1,1]^{n}$ into $[-1,1]$. For $f,g\in \mathcal{F}$, let $f*g\in \mathcal{F}$ be defined at $z\in \{-1,1\}^n$ by $(f*g)(z):=\E_{x\sim U_n}[f(x)g(x\cdot z)].$ It can be shown that for all $S\subseteq[n]$, we have $\widehat{f*g}(S)=\widehat{f}(S)\widehat{g}(S)$.

\paragraph*{Level-$k$ Fourier Mass} For $f\in \mathcal{F}$ and $k\in \{0,\ldots,n\}$, let $L_k(f):=\sum_{|S|=k} |\widehat{f}(S)|$ denote the level-$k$ Fourier mass of $f$. For a family $\mathcal{H}\subseteq \mathcal{F}$ of Boolean functions, let $L_k(\mathcal{H}):=\sup_{H\in \mathcal{H}}L_k(H)$.

\subsection{The Forrelation Problem}

Let $k,N\in \mathbb{N}$ be parameters, where $N=2^n$ for some $n\in \mathbb{N}$. We assume that $k=o(N^{1/50})$. Fix a parameter $\epsilon=\frac{1}{60k^2 \ln N}$. Let $\mathcal{U}$ refer to $U_{2N}$.

\paragraph*{Hadamard Matrix}
The Hadamard matrix $H_N$ of size $N$ is an $N\times N$ matrix. The rows and columns are indexed by strings $a$ and $b$ respectively where $a,b\in \{0,1\}^n$ and the $(a,b)$-th entry of $H_N$ is defined to be $\frac{1}{\sqrt{N}} (-1)^{\left<a,b\right>_2}$. Equivalently,
\[ H_N(a,b):=\begin{cases} \frac{-1}{\sqrt{N}} &\text{ if } \sum_{i=1}^n a_ib_i\equiv {1 \mod  2} \\
 \frac{+1}{\sqrt{N}} & \text{ if } \sum_{i=1}^n a_ib_i\equiv {0 \mod  2}   \end{cases}  \]

\paragraph*{The Forrelation Function}
The Forrelation Function $forr:\mathbb{R}^{2N}\rightarrow\mathbb{R}$ is defined as follows.
Let $z\in \mathbb{R}^{2N}$ and $x,y\in \mathbb{R}^N$ be such that $z=(x,y)$. Then, \[ forr(z) := \frac{1}{N} \langle x , H_N y \rangle  \]

\paragraph*{The $\oplus^k$ Forrelation Decision Problem}
\begin{definition} [The $\oplus^k$ Forrelation Decision Problem]
\label{forrelationproblem}
The Forrelation Decision Problem is the partial Boolean function $F:\{-1,1\}^{2N}\rightarrow \{-1,1\}$ defined as follows. For $z\in \{-1,1\}^{2N}$, let
\[F(z):=\begin{cases} -1 & \text{ if } forr(z ) \ge \epsilon/2 \\ 1 & \text{ if } forr(z) \le \epsilon/4\\  \text{undefined} & \text{ otherwise } \end{cases} \]
The $\oplus^k$ Forrelation Decision Problem $F^{(k)}:\{-1,1\}^{2kN}\rightarrow \{-1,1\}$ is defined as the \textsc{xor} of $k$ independent copies of $F$. To be precise, for every $z_1,\ldots,z_k\in \{-1,1\}^{2N}$, let
\[ F^{(k)}(z_1,\ldots,z_k):=\prod_{j=1}^k F(z_j)\]
\end{definition}

\paragraph*{The Gaussian Forrelation Distribution $\mathcal{G}$}
\begin{definition}\label{gaussian} Let $\mathcal{G}$ denote the Gaussian distribution over $\mathbb{R}^{2N}$ defined by the following process.
\begin{enumerate}
\item Sample $x_1,\ldots,x_N\sim \mathcal{N}(0,\epsilon)$ independently.
\item Let $x=(x_1,\ldots,x_N)$ and $y=H_N x$.
\item Output $(x,y)$.
\end{enumerate}
\end{definition}
The distribution $\mathcal{G}$ can be equivalently expressed as $\mathcal{N}\left(0, \epsilon \begin{bmatrix} \mathbb{I}_N & H_N \\ H_N & \mathbb{I}_N \end{bmatrix}\right)$.

\paragraph*{Moments of $\mathcal{G}$} We state some useful facts about the moments of $\mathcal{G}$. We use the following notation to refer to the moments of $\mathcal{G}$. For subsets $S,T\subseteq [N]$, let
$\widehat{\mathcal{G}}(S,T):= \underset{(x,y)\sim \mathcal{G}}{\E} \left[ \prod_{i\in S}x_i \prod_{j\in T}y_j \right] $. The following claim and its proof appear as Claim 4.1 in \cite{raztal}. We omit the proof.
\begin{claim} \label{claim1} Let $S,T\subseteq [N]$ and $i,j\in [N]$. Let $i_1=|S|,i_2=|T|$. Then,
\begin{enumerate}
\item $\widehat{\mathcal{G}}(\{i\},\{j\})=\epsilon N^{-1/2}(-1)^{\left<i,j\right>_2}$.
\item $\widehat{\mathcal{G}}(S,T)=0$ if $i_1\neq i_2$.
\item $\left|\widehat{\mathcal{G}}(S,T)\right|\le \epsilon ^{i} i!N^{-i/2}$ if $i=i_1=i_2$.
\end{enumerate}
\end{claim}

\subsection{Hard Distributions over $\mathbb{R}^{2kN}$}
Let $\mathcal{P},\mathcal{Q}$ be two probability distributions on the domain $\mathbb{D}:=\mathbb{R}^{2N}$. Let $S\subseteq [k]$. We define $ \mathcal{P}^S \mathcal{Q}^{\bar{S}}$ to be the distribution on $\mathbb{D}^k$ defined by sampling $x=(x_1,\ldots,x_k)$ where $x_1,\ldots,x_k\in \mathbb{D}$ are sampled as follows.
\[ \text{For each $j\in [k]$, independently sample } \begin{cases} x_j \sim \mathcal{P}& \text{ if }j\in S\\
x_j\sim \mathcal{Q} & \text{ if } j\in\bar{S} \end{cases} \]
Note that for every $I=(I_1,\ldots,I_k)\subseteq [2kN]$, where $I_1,\ldots,I_k\subseteq[2N]$, we have the following.
\[ \widehat{\mathcal{P}^S\mathcal{Q}^{\bar{S}} }(I)=\prod_{j\in S} \widehat{\mathcal{P}}(I_j) \cdot \prod_{j\notin S}\widehat{\mathcal{Q}}(I_j) \]
\begin{definition}\label{harddistributions}
Let $\mathcal{G}$ be the distribution in \cref{gaussian} and $\mathcal{U}=U_{2N}$. Define a pair of distributions $\mu_0^{(k)},\mu_1^{(k)}$ on $\mathbb{R}^{2kN}$ as follows.
\[ \mu_0^{(k)}:= \frac{1}{2^{k-1}}\sum_{\substack{S\subseteq [k]\\ |S|\text{ is even }}}\mathcal{G}^S \mathcal{U}^{\bar{S}}\quad\quad\text{ and }\quad\quad\mu_1^{(k)}:= \frac{1}{2^{k-1}}\sum_{\substack{S\subseteq [k]\\ |S|\text{ is odd }}}\mathcal{G}^S \mathcal{U}^{\bar{S}}\]
\end{definition}
\begin{lemma} \label{moments} Let $I=(I_1,\ldots,I_k)\subseteq [2kN]$, where each $I_j \subseteq [2N]$.
\begin{enumerate}
\item If $|I|<2k$ or if $I_j=\emptyset$ for some $j\in [k]$, then  $\widehat{\mu_0^{(k)}}(I)=\widehat{\mu_1^{(k)}}(I)$.
\item If $|I_j|$ is odd for some $j\in [k]$, then $\widehat{\mu_0^{(k)}}(I)=\widehat{\mu_1^{(k)}}(I)$.
\item Let $|I|=2i$ for some $i\in \mathbb{N}$. Then,
$\left|\widehat{\mu_0^{(k)}}(I)-\widehat{\mu_1^{(k)}}(I)\right| \le 2^{-k+1} \epsilon^{i}N^{-i/2}i! $.
\end{enumerate}
\end{lemma}

\begin{proof}[Proof of \cref{moments}.]
Note that we have the following equality
\begin{align}
\begin{split}\label{differencemoments}
\widehat{\mu_0^{(k)}}(I) - \widehat{\mu_1^{(k)}}(I)&\triangleq\frac{1}{2^{k-1}} \left( \sum_{\substack{S\subseteq [k]\\ |S|\text{ is even }}}\widehat{\mathcal{G}^S \mathcal{U}^{\bar{S}} }(I)- \sum_{\substack{S\subseteq [k]\\ |S|\text{ is odd }}}\widehat{\mathcal{G}^S \mathcal{U}^{\bar{S}}}(I)\right)  \\
&= \frac{1}{2^{k-1}}  \left( \sum_{S\subseteq [k]} (-1)^{|S|} \widehat{\mathcal{G}^S\mathcal{U}^{\bar{S}}}(I)\right)  \\
&= \frac{1}{2^{k-1}}   \left( \sum_{S\subseteq [k]} (-1)^{|S|} \prod_{j\in S} \widehat{\mathcal{G}}(I_j) \prod_{j\notin S} \widehat{\mathcal{U}}(I_j)\right)  \\
& = \frac{1}{2^{k-1}}  \prod_{j=1}^{k} \left(\widehat{\mathcal{U}}(I_j) -\widehat{\mathcal{G}}(I_j)\right)
\end{split}
\end{align}
\begin{enumerate}[(1.)]
\item If $|I|<2k$ then there exists some $j\in [k]$ such that $|I_j|<2$. If $|I_j|=0$, then $\widehat{\mathcal{G}}(I_j)=\widehat{\mathcal{U}}(I_j)=1$. If $| I_j|=1$, \cref{claim1} implies that $\widehat{\mathcal{G}}(I_j)=\widehat{\mathcal{U}}(I_j)=0$. This along with \cref{differencemoments} implies that $\widehat{\mu_0^{(k)}}(I)=\widehat{\mu_1^{(k)}}(I)$.
\item Suppose $|I_j|$ is odd for some $j\in [k]$. \cref{claim1} implies that $\widehat{\mathcal{G}}(I_j)=\widehat{\mathcal{U}}(I_j)=0$. This, along with \cref{differencemoments} implies that $\widehat{\mu_0^{(k)}}(I)=\widehat{\mu_1^{(k)}}(I)$.
\item Due to item (1.) and (2.) of this lemma, we may assume that $I_j\neq \emptyset$ and $|I_j|$ is even for every $j\in [k]$, otherwise $\widehat{\mu_0^{(k)}}(I)-\widehat{\mu_1^{(k)}}(I)=0$ and the inequality is trivially true. For each $j\in [k]$, let $|I_j|=2i_j$ for some $i_j\in \mathbb{N}$. \cref{claim1} states that if $|I_j|=2i_j$, then $|\widehat{\mathcal{G} }(I_j) |\le \epsilon^{i_j} i_j!N^{-i_j/2}$. Since $I_j\neq \emptyset$, we have $\widehat{\mathcal{U}}(I_j)= 0$. This, along with \cref{differencemoments} implies that
\begin{align*}
\begin{split}
\left|\widehat{\mu_0^{(k)}}(I_1,\ldots,I_k) - \widehat{\mu_1^{(k)}}(I_1,\ldots,I_k)\right|& =\frac{1}{2^{k-1}}   \left|  \prod_{j=1}^{k} \left(\widehat{\mathcal{G}}(I_j) -\widehat{\mathcal{U}}(I_j)\right)\right| \\
&\le \frac{1}{2^{k-1}}  \prod_{j=1}^{k} \epsilon^{i_j}  i_j!N^{-i_j/2}\\
&=  \frac{1}{2^{k-1}}  \epsilon^{i} N^{-i/2}  \prod_{j=1}^{k} i_j! \le 2^{-k+1} \epsilon^{i} N^{-i/2} i!\\
\end{split}
\end{align*}
\end{enumerate}
This completes the proof of \cref{moments}.
\end{proof}

\subsection{Rounding Distributions to the Boolean Hypercube}

Let $trnc:\mathbb{R}\rightarrow [-1,1]$ denote the truncation function, whose action on $a\in \mathbb{R}$ is given by
\[ trnc(a)=\begin{cases} sign(a) & \text{if } a\notin [-1,1] \\ a&\text{otherwise} \end{cases}\]
For $l\in \mathbb{R}$, we also use $trnc:\mathbb{R}^l\rightarrow[-1,1]^l$ to refer to the function that applies the above truncation function coordinate-wise.

\begin{definition}\label{rounding}
Let $\mu$ be any distribution on $\mathbb{R}^{M}$. We define the rounded distribution $\tilde{\mu}$ on $\{-1,1\}^{M}$ as follows.
\begin{enumerate}
\item Sample $z\sim \mu$.
\item For each coordinate $i\in [M]$, independently, let $z'_i=1$ with probability $\frac{1+trnc(z_i)}{2}$ and $z'_i=-1$ with probability $\frac{1-trnc(z_i)}{2}$.
\item Output $z'=(z'_1,\ldots,z'_M)$.
\end{enumerate}
Let $z_0\in \mathbb{R}^M$ and $\mu$ be the distribution whose support is $\{z_0\}$. We use $\tilde{z}_0$ to refer to $\tilde{\mu}$.
\end{definition}

We show some useful facts about expectations of multilinear functions over these distributions.

\begin{claim} \label{multilinearity} Let $H:\mathbb{R}^M\rightarrow \mathbb{R}$ be any multilinear polynomial and $a\in \mathbb{R}^M$. Let $\mu$ be a distribution on $\mathbb{R}^M$ where each coordinate is sampled independently of the rest so that $\E_{z\sim \mu}[z]=a$. Then,
\[ \E_{z\sim \mu}[ H(z)] = H(a)\]
\end{claim}

\begin{corollary} \label{claim2}  Let $H:\mathbb{R}^{M}\rightarrow \mathbb{R}$ be any multilinear polynomial. Let $\mu$ be any distribution on $\mathbb{R}^M$ and $\tilde{\mu}$ be the distribution on $\{-1,1\}^{M}$ obtained by rounding $\mu$ as in \cref{rounding}. Then,
\[ \underset{z \sim\tilde{\mu}}{\E}[H(z)]= \underset{z\sim \mu}{\E}[H(trnc(z))]  \]
\end{corollary}

\begin{claim}\label{claim5} Let $H:\mathbb{R}^{2kN}\rightarrow\mathbb{R}$ be any multilinear polynomial mapping $\{-1, 1\}^{2kN}$ into $[-1,1]$. Let $z_0$ and $P$ be in $[-1/2,1/2]^{2kN}$. Then,
\[ \E_{z\sim \mathcal{G}^{(k)}} \left[\left| H(trnc(z_0+P\cdot z))-H(z_0+P\cdot z) \right|\right] \le O\left( \frac{1}{N^{5k^2}}\right) \]
\end{claim}

\begin{proof}[Proof of \cref{multilinearity}] Let $T\subseteq[M]$ and $z\sim\mu$. The given assumption on $\mu$ is that each $z_j$ for $j\in [M]$ is sampled independently so that $\E_{z\sim \mu}[z_j]=a_j$. This implies that $\E_{z\sim \mu}[\chi_T(z)]\triangleq  \E_{z\sim \mu} \left[ \prod_{j\in T}z_j\right] = \prod_{j\in T} a_j  \triangleq\chi_T(a)$. Note that the quantities $\E_{z\sim \mu}[H(z)]$ and $H(a)$ are both linear with respect to $H$. Since we have shown that  $\E_{z\sim \mu}[H(z)]=H(a)$  for all character functions $H$, this observation implies that $\E_{z\sim \mu}[H(z)]=H(a)$ for all multilinear functions $H$.
\end{proof}

\begin{proof}[Proof of \cref{claim2} from \cref{multilinearity}] Observe that for every $z\in \mathbb{R}^M$, the  distribution $\tilde{z}$ as in \cref{rounding} satisfies the hypothesis in \cref{multilinearity} with $a=trnc(z)$. \cref{multilinearity} implies that $\underset{z'\sim \tilde{z}}{\E}[H(z')]=  H(trnc(z))$. Therefore, $ \underset{z\sim \tilde{\mu}}{\E}[H(z)]\triangleq \underset{z\sim \mu}{\E} \hspace{0.1cm} \underset{z'\sim \tilde{z}}{\E}[H(z') \mid z]= \underset{z\sim \mu}{\E} [H(trnc(z))]$. \end{proof}

\cref{claim2} is similar to Equation (2) from \cite{raztal} and Claim 2.2 from \cite{grt}. \cref{claim5} is similar to Claim 5.3 from \cite{raztal}. The proof of this is also identical, so we omit it. We remark that the bound in \cite{raztal} is $8\cdot N^{-2}$ as opposed to our bound of $O\left( N^{-5k^2}\right).$ This difference in parameters arises from our choice of $\epsilon=\frac{1}{60k^2\ln N}$ as opposed to their choice of $\epsilon=\frac{1}{24\ln N}$. We also remark that the claim as stated in \cite{raztal} is for scalars $P\in [-1/2,1/2]$ as opposed to our assumption of $P\in [-1/2,1/2]^{2kN}$. However, their proof works under this assumption as well.

\subsection{The Forrelation Distribution}

Let $k\in\mathbb{N}$. Let $\tilde{\mu}_0^{(k)}$ and $\tilde{\mu}_1^{(k)}$ (respectively $\mathcal{\tilde{G}}$) be distributions over $\{-1, 1\}^{2kN}$ (respectively $\{-1,1\}^{2N}$) generated from rounding $\mu_1^{(k)}$ and $\mu_0^{(k)}$ (respectively $\mathcal{G}$) according to \cref{rounding}. Observe that we may alternatively define  $\tilde{\mu}_0^{(k)}$ and $\tilde{\mu}_1^{(k)}$ as follows.
\begin{definition}\label{roundingharddistributions} Let $\mathcal{G}$ be as in \cref{gaussian} and $\mathcal{U}=U_{2N}$. Let
\[ \tilde{\mu}_0^{(k)}:= \frac{1}{2^{k-1}}\sum_{\substack{S\subseteq [k]\\ |S|\text{ is even }}}\tilde{\mathcal{G}}^S {\mathcal{U}}^{\bar{S}}\quad\quad\text{ and } \quad\quad \tilde{\mu}_1^{(k)}:= \frac{1}{2^{k-1}}\sum_{\substack{S\subseteq [k]\\ |S|\text{ is odd }}}\tilde{\mathcal{G}}^S {\mathcal{U}}^{\bar{S}}\]
We refer to $\tilde{\mu}_1^{(1)}\triangleq \tilde{\mathcal{G}}$ as the Forrelation Distribution.
\end{definition}

We show that the distributions $\tilde{\mu}_1^{(k)}$ and $\tilde{\mu}_0^{(k)}$ put considerable mass on the \textsc{yes} and \textsc{no} instances of $F^{(k)}$, respectively, where $F^{(k)}$ is the $\oplus^k$ Forrelation Decision Problem as in \cref{forrelationproblem}.

\begin{lemma}\label{concentrationcorollary2} Let $\tilde{\mu}_0^{(k)}$ and $\tilde{\mu}_1^{(k)}$ be distributions as in \cref{roundingharddistributions} and $F^{(k)}$ be the $\oplus^k$ Forrelation Decision Problem as in \cref{forrelationproblem}. Then,
\[ \underset{z\sim\tilde{\mu}_0^{(k)}}{\p}[ F^{(k)}(z)=1 ] \ge 1 - O\left( \frac{k}{N^{6k^2}} \right) \quad\text{ and }\quad \underset{z\sim \tilde{\mu}_1^{(k)} }{\p}[ F^{(k)}(z)=-1 ] \ge 1-  O\left( \frac{k}{N^{6k^2}} \right) \]
\end{lemma}

The proofs of these use hypercontractivity to show concentration inequalities for low degree polynomials under product distributions on the Boolean hypercube. These proofs are technical and are deferred to the appendix.

\subsection{Closure under Restrictions}

\begin{definition} \label{restriction} Let $a\in \{-1,1,0\}^M$. Let $\rho_a:\mathbb{R}^{M} \rightarrow\mathbb{R}^{M}$ be a restriction defined as follows. For $v\in \mathbb{R}^M$, let $\rho_a(v) \in \mathbb{R}^M$ be such that for all $j\in [M]$,
\[ (\rho_a(v)) (j):=\begin{cases} v(j) & \text{if } a(j)=0\\ a(j) & \text{otherwise} \end{cases} \]
For a function $F:\{-1,1\}^M\rightarrow \mathbb{R}$, the restricted function $F\circ\rho_v:\{-1,1\}^M\rightarrow\mathbb{R}$ is defined at $z\in \{-1,1\}^M$ by $(F\circ \rho_v)(z) := F(\rho_v(z))$.

We say that a family $\mathcal{H}$ of Boolean functions in $M$ variables is {\it closed under restrictions} if for all restrictions $v\in \{-1,1,0\}^M$ and $H\in\mathcal{H}$, the  restricted function $H\circ \rho_v$ is in $\mathcal{H}$.
\end{definition}

\section{The Main Result}

Let $N\in \mathbb{N}$ be a parameter describing the input size. We will assume that $N$ is a sufficiently large power of 2. Let $k\in \mathbb{N}$. We assume that $k=o(N^{1/50})$. Let $\epsilon=\frac{1}{60k^2 \ln N}$ be the parameter defining $\mathcal{G}$ as before.

\begin{theorem} \label{theorem1}  Let $\mathcal{H}$ be a family of Boolean functions on $2kN$ variables, each of which maps $\{-1,1\}^{2kN}$ into $[-1,1]$. Assume that $\mathcal{H}$ is closed under restrictions. Let $\tilde{\mu}_0^{(k)},\tilde{\mu}_1^{(k)}$ be the distributions over $\{-1,1\}^{2kN}$ as in \cref{roundingharddistributions}. Then, for every $H\in \mathcal{H}$,
\[ \left| \underset{ z\sim \tilde{\mu}_0^{(k)}}{\E}\left[ H(z) \right]  -  \underset{ z\sim \tilde{\mu}_1^{(k)}}{\E}\left[ H(z) \right] \right| \le O\left( \frac{L_{2k}(\mathcal{H})}{N^{k/2}}\right) + o\left(  \frac{1}{N^{k/2}}\right)  \]
\end{theorem}

\begin{definition}\label{realharddistributions} Let $\tilde{\mu}_0^{(k)}, \tilde{\mu}_1^{(k)}$ be as in \cref{roundingharddistributions}. Let $\sigma_0^{(k)}$ (respectively $\sigma_1^{(k)}$) be obtained by conditioning $\tilde{\mu}_0^{(k)}$ on being a \textsc{no} (respectively \textsc{yes}) instance of $F^{(k)}$.
\end{definition}

\begin{corollary} \label{restatetheorem1} Under the same hypothesis as \cref{theorem1}, for every $H\in \mathcal{H}$
\[ \left| \underset{ z\sim \sigma_0^{(k)}}{\E}\left[ H(z) \right]  -  \underset{ z\sim \sigma_1^{(k)}}{\E}\left[ H(z) \right] \right| \le O\left( \frac{L_{2k}(\mathcal{H})}{N^{k/2}}\right) + o\left(  \frac{1}{N^{k/2}}\right)  \]
\end{corollary}

\subsection{Applications to Quantum versus Classical Separations}

\paragraph*{Query Complexity Separations}
\begin{lemma} \label{corollary3oftheorem1} Let $D:\{-1,1\}^{2kN}\rightarrow \{-1,1\}$ be a deterministic decision tree of depth $d\ge 1$. Then,
\[ \left| \underset{ z\sim \sigma_0^{(k)}}{\E}\left[ D(z) \right]  -  \underset{ z\sim \sigma_1^{(k)}}{\E}\left[ D(z) \right] \right| \le \left(\frac{ O\left(  d\log (kN)\right)}{N^{1/2}} \right)^k   \]
\end{lemma}

\begin{theorem} \label{theorem4} $F^{(k)}$ can be computed in the bounded-error quantum query model with $O(k ^5  \log^2 N\log k)$ queries. However, every randomized decision tree of depth $\tilde{o}(\sqrt{N})$ has a worst-case success probability of at  most $\frac{1}{2}+\exp(-\Omega(k))$.
\end{theorem}

Setting $k=\lceil \log^c N\rceil$ for $c\in \mathbb{N}$ in \cref{theorem4} gives us an explicit family of partial functions that are computable by quantum query algorithms of cost $\tilde{O}(\log^{5c+2} N)$, however every randomized query algorithm of cost $\tilde{o}(N^{\frac{1}{2}})$ has at most $\frac{1}{2^{\Omega(\log^c N)}}$ advantage over random guessing.

\paragraph*{Communication Complexity Separations}
\begin{definition} [The $\oplus^k$ Forrelation Communication Problem $F^{(k)}\circ \textsc{xor}$]
Alice is given $x$ and Bob is given $y$ where $x,y \in\{-1, 1\}^{2kN}$. Let $F^{(k)}$ be as in \cref{forrelationproblem}. Their goal is to compute the partial function $ F^{(k)}(x\cdot y)$.
\end{definition}

\begin{lemma} \label{corollary2oftheorem1} Let $C:\{-1,1\}^{2kN}\times \{-1,1\}^{2kN}\rightarrow \{-1,1\}$ be any deterministic protocol of communication complexity $c$. Then,
\[ \left| \underset{\substack{x \sim U_{2kN}\\ z\sim \sigma_0^{(k)}}}{\E}\left[ C(x,x\cdot z) \right]  -  \underset{\substack{x \sim U_{2kN}\\ z\sim \sigma_1^{(k)}}}{\E}\left[ C(x,x\cdot z) \right]  \right| \le O\left( \frac{(c+8k)^{2k}}{N^{k/2}}\right)   \]
\end{lemma}

\begin{theorem} \label{theorem3} $F^{(k)}\circ \textsc{xor}$ can be solved in the quantum simultaneous with entanglement model with $O(k^5 \log^3 N \log k)$ bits of communication, when Alice and Bob share $O(k^5 \log^3 N \log k)$ EPR pairs. However, any randomized protocol of cost $\tilde{o}(N^{1/4})$ has a worst-case success probability of at  most $\frac{1}{2}+\exp(-\Omega(k))$.
\end{theorem}

Setting $k=\lceil \log^c N\rceil$ for $c\in \mathbb{N}$ in \cref{theorem3} gives us an explicit family of partial functions that are computable by quantum simultaneous protocols of cost $\tilde{O}(\log^{5c+3} N)$ when Alice and Bob share $\tilde{O}(\log^{5c+3} N)$ EPR pairs, however every interactive randomized protocol of cost $\tilde{o}(N^{\frac{1}{4}})$ has at most $\frac{1}{2^{\Omega(\log^c N)}}$ advantage over random guessing.

\paragraph*{Circuit Complexity Separations}
\begin{lemma} \label{corollary1oftheorem1} Let $C:\{-1,1\}^{2kN}\rightarrow \{-1,1\}$ be an AC0 circuit of depth $d\ge 1$ and size $s$. Then,
\[ \left| \underset{ z\sim \sigma_0^{(k)}}{\E}\left[ C(z) \right]  -  \underset{ z\sim \sigma_1^{(k)}}{\E}\left[ C(z) \right] \right| \le \left( \frac{O\left(\log^{2d-2}(s)\right)}{N^{1/2}}\right)^k   \]
\end{lemma}

\begin{theorem} \label{theorem2} The distributions $\sigma_1^{(k)}$ and $\sigma_0^{(k)}$ can be distinguished by a bounded-error quantum query protocol with $O(k ^5  \log^2 N\log k)$ queries with $2/3$ advantage. However, every constant depth circuit of size $o\left( \exp\left( N^{\frac{1}{4(d-1)}}\right) \right)$ can distinguish these distributions with at most $\exp(-\Omega(k))$ advantage.
\end{theorem}

Setting $k=\lceil \log^c N\rceil$ for $c\in \mathbb{N}$ in \cref{theorem2} gives us an explicit family of distributions that are distinguishable by cost $\tilde{O}(\log^{5c+2} N)$ quantum query algorithms, however every constant depth circuit of quasipolynomial size can distinguish them with at most $\frac{1}{2^{\Omega(\log^c N)}}$ advantage.

\section{Single Step Analysis Around the Origin}

\begin{lemma} \label{mainlemma1} Let $H$ be a Boolean function on $2kN$ variables that maps $\{-1, 1\}^{2kN}$ into $[-1, 1]$. Let $p\le \frac{1}{2N}$ and $P\in [-p,p]^{2kN}$. Then,
\[ \Delta: = \left| \underset{z\sim P\cdot \mu_0^{(k)} }{\E}[H(z)]- \underset{z\sim P\cdot \mu_1^{(k)} }{\E} [H( z) ]\right| \le   O\left( 2^{-2k}\cdot \frac{L_{2k}(H)  p^{2k}}{N^{k/2}}   +  p^{2(k+1)} N^{(k+1)/2} \right)    \]
\end{lemma}

\begin{proof}[Proof of \cref{mainlemma1}]

For all $z\in \mathbb{R}^{2kN}$, we have $H(z)=\sum_{S\subseteq [2kN]} \widehat{H}(S)\underset{i\in S}{\prod} z_i$. This implies that
\begin{align*}\begin{split}
\Delta&= \left| \underset{S\subseteq [2kN]}{\sum} \widehat{H}(S) \left( \underset{z\sim P\cdot \mu_0^{(k)}}{\E}\Big[\prod_{i\in S}z_i\Big]  -\underset{z\sim P\cdot \mu_1^{(k)}}{\E}\Big[\prod_{i\in S}z_i\Big]\right) \right| \\
&= \left| \underset{S\subseteq [2kN]}{\sum} \widehat{H}(S) \cdot \prod_{i\in S}P_i \cdot  \left( \underset{z\sim \mu_0^{(k)}}{\E}\Big[\prod_{i\in S}z_i\Big]  -\underset{z\sim \mu_1^{(k)}}{\E}\Big[\prod_{i\in S}z_i\Big]\right) \right| \\
&= \left| \underset{S\subseteq [2kN]}{\sum} \widehat{H}(S) \cdot \prod_{i\in S}P_i \cdot \left( \widehat{\mu_0^{(k)}}(S)  -\widehat{\mu_1^{(k)}}(S) \right) \right| \\
&\le \underset{S\subseteq [2kN]}{\sum} |\widehat{H}(S)|\cdot p^{|S|} \cdot \left| \widehat{ \mu_0^{(k)}}(S)  -\widehat{ \mu_1^{(k)}}(S) \right| \quad\quad\ldots\text{ since $P\in [-p,p]^{2kN}$}
\end{split}\end{align*}

We now apply \cref{moments} to bound the difference in moments between the distributions $\mu_1^{(k)}$ and $\mu_0^{(k)}$. \cref{moments} implies that if $|S| <2k$ or $|S|$ is odd, then $\widehat{\mu_0^{(k)}}(S)=\widehat{\mu_1^{(k)}}(S)$. Furthermore, if $| S|=2i$ for some $i\in \mathbb{N}$, then $\left| \widehat{ \mu_0^{(k)}}(S)-\widehat{ \mu_1^{(k)}}(S) \right| \le 2^{-k+1} \epsilon^i N^{-i/2}  i! $. This implies that
\[\Delta \le \sum_{i=k}^{kN}  \left( \sum_{|S|=2i} |\widehat{H}(S)| \right)\cdot 2^{-k+1}  \epsilon^i N^{-i/2}   i! p^{2i}  \]
Since $H$ maps $\{-1,1\}^{2kN}$ to $[-1,1]$, we can bound $\sum_{|S|=2i} |\widehat{H}(S)|$ by $\sqrt{2kN \choose 2i}$.~\footnote{This is because $\sum_{|S|=2i} |\widehat{H}(S)|\le \sqrt{\sum_{|S|=2i}1}\sqrt{\sum_{|S|=2i}\widehat{H}(S)^2}\le \sqrt{{2kN \choose 2i}}$.} We also bound $2^{-k+1}$ by 1. This, along with the previous inequality implies that
\[ \Delta \le  L_{2k}(H) \cdot (2^{-k+1} k! \epsilon^k) \cdot N^{-k/2} p^{2k} +   \sum_{i=k+1}^{kN} \left(\sqrt{2kN\choose 2i}\cdot i! \epsilon^i\right) \cdot N^{-i/2}  p^{2i}    \]
Note that $  \sqrt{2kN\choose 2i}\cdot i! \le \frac{(2k)^iN^i}{\sqrt{(2i)!}} i! =O\left( \frac{(2k)^ie^{i}}{(2i)^{i}} \cdot \frac{i^i}{e^i} \cdot N^i   \right) =O\left( k^i N^i \right)$.
Furthermore, since $\epsilon=\frac{1}{60k^2\ln N}$, for all $i\ge k$, we have $   \epsilon^i k^i= O\left( \frac{1}{60^k k^{2k}} \cdot k^k \right) = O\left(\frac{1}{2^k}\right)$. This implies that
$2^{-k+1} k!\epsilon^k = O\left( 2^{-2k} \right)$ and $\sqrt{2kN \choose 2i} \cdot i! \epsilon^i =O( N^i)$. Substituting these bounds in the previous inequality for $\Delta$, we have
\[ \Delta \le O\left( 2^{-2k}\cdot  L_{2k}(H) N^{-k/2} p^{2k} +  \sum_{i=k+1}^{kN} N^{i/2} p^{2i} \right)   \]
In the summation $\sum_{i \ge k+1}N^{i/2} p^{2i} $, every successive term is smaller than the previous by a factor of at least $1/4$. This is because the assumption $p\le \frac{1}{2N}$ implies that $N^{1/2}p^2\le \frac{1}{4}$. Thus, we can bound this summation by twice the first term, which is  $O(N^{(k+1)/2}p^{2(k+1)} ).$ This implies that
\[\Delta \le O \left(  2^{-2k}\cdot L_{2k}(H) N^{-k/2} p^{2k} +N^{(k+1)/2}p^{2(k+1)} \right) \]
This completes the proof of \cref{mainlemma1}.
\end{proof}

\section{Single Step Analysis Away from the Origin}

\begin{lemma} \label{mainlemma2} Let $\mathcal{H}$ be a family of Boolean functions on $2kN$ variables, each of which maps $\{-1,1\}^{2kN}$ into $[-1,1]$. Assume that $\mathcal{H}$ is closed under restrictions. Let $p\le \frac{1}{4N}$  and $z_0\in [-1/2,1/2]^{2kN}$. Then, for all $H\in \mathcal{H}$,
\[ \Delta: = \left| \underset{z\sim p\cdot \mu_0^{(k)} }{\E}[H(z_0+z)]- \underset{z\sim p\cdot \mu_1^{(k)} }{\E} [H(z_0+ z) ]\right|  \le   O\left(2^{-2k}\cdot  \frac{L_{2k}(\mathcal{H})  (2p)^{2k}}{N^{k/2}}   +  (2p)^{2(k+1)} N^{(k+1)/2} \right)    \]
\end{lemma}

Let $\mathcal{O}$ denote the distribution on $\mathbb{R}^{2N}$ whose support is $\{0\}$ (i.e, the distribution that puts all its mass on the zero vector in $\mathbb{R}^{2N}$).

\begin{corollary} \label{maincorollary2} Under the same hypothesis as \cref{mainlemma2}, for all $H\in \mathcal{H}$,
\[ \Delta: = \frac{1}{2^{k-1}}\left| \sum_{S\subseteq [k]}(-1)^{|S|} \underset{\substack{z\sim z_0+\\ p\cdot \mathcal{G}^S\mathcal{O}^{\bar{S}}}}{\E}  \left[ H(z) \right]   \right| \le   O\left(2^{-2k}\cdot  \frac{L_{2k}(\mathcal{H})  (2p)^{2k}}{N^{k/2}}   +  (2p)^{2(k+1)} N^{(k+1)/2} \right)    \]
\end{corollary}

\begin{proof}[Proof of \cref{maincorollary2} from \cref{mainlemma2}] We show that the expressions for $\Delta$ in \cref{maincorollary2} and \cref{mainlemma2} are identical. Let $ \Gamma:=   \left| \underset{z\sim p\cdot \mu_0^{(k)} }{\E}[H(z_0+z)]- \underset{z\sim  p\cdot \mu_1^{(k)} }{\E}[ H( z_0+z)]\right| $ be the expression for $\Delta$ in \cref{mainlemma2}. By the definition of $\mu_0^{(k)},  \mu_1^{(k)}$ as in \cref{harddistributions}, we have
\begin{equation}\label{maincorollary2eqn} \Gamma =\frac{1}{2^{k-1}}\left|   \sum_{S\subseteq [k] } (-1)^{|S|} \underset{z\sim p\cdot \mathcal{G}^S\mathcal{U}^{\bar{S}}}{\E} \left[H(z_0+z )  \right]\right| =\frac{1}{2^{k-1}}\left|   \sum_{S\subseteq [k] } (-1)^{|S|} \underset{z\sim  \mathcal{G}^S\mathcal{U}^{\bar{S}}}{\E} \left[H(z_0+pz )  \right]\right| \end{equation}
Let $S\subseteq[k]$. We now show that $\E_{z\sim \mathcal{G}^S\mathcal{U}^{\bar{S}}} [H(z_0+pz)]=\E_{z\sim \mathcal{G}^S\mathcal{O}^{\bar{S}} } [H(z_0+pz)].$ Substituting this in the above equation would complete the proof. Let $z_1\sim \mathcal{G}^S\mathcal{O}^{\bar{S}}$ and $z_2\sim \mathcal{O}^S\mathcal{U}^{\bar{S}}$. Note that $z_1+z_2\sim \mathcal{G}^S\mathcal{U}^{\bar{S}}$. Fix $z_1\in \mathbb{R}^{2kN}$. Note that the multilinear polynomial $H(z_0+pz_1+pz_2)$ over $z_2$ and the distribution $ \mathcal{O}^S\mathcal{U}^{\bar{S}}$ satisfies the hypothesis in \cref{multilinearity} for $a=0$. \cref{multilinearity} implies that for all $z_1\in \mathbb{R}^{2kN}$, we have $\E_{z_2\sim \mathcal{O}^S\mathcal{U}^{\bar{S}}}[H(z_0+pz_1+pz_2)\mid z_1]=H(z_0+pz_1).$ It then follows that
\[ \underset{z\sim \mathcal{G}^S\mathcal{U}^{\bar{S}}}{\E} [H(z_0+pz)] = \underset{\substack{z_1\sim \mathcal{G}^S\mathcal{O}^{\bar{S}}\\ z_2\sim \mathcal{O}^S\mathcal{U}^{\bar{S}}}}{\E}[H(z_0+pz_1+pz_2)]= \underset{z_1\sim \mathcal{G}^S\mathcal{O}^{\bar{S}}}{\E} [H(z_0+pz_1) ]=\underset{\substack{z\sim z_0+\\ p\cdot \mathcal{G}^S\mathcal{O}^{\bar{S}}}}{\E}  \left[ H(z) \right]  \]
Substituting the above in \cref{maincorollary2eqn} implies that $\Delta=\Gamma$. This, along with \cref{mainlemma2} completes the proof of \cref{maincorollary2}.
\end{proof}

\begin{proof}[Proof of \cref{mainlemma2}]
Let $v\in \{-1,1,0\}^{2kN}$ be obtained by the following process, which we denote by $v\sim z_0$. For every $i\in [2kN]$, independently, set
\[ v(i) :=\begin{cases} sign(z_0(i)) &\text{ with probability } |z_0(i)| \\
 0& \text{ with probability } 1-|z_0(i)| \end{cases} \]
Let $\rho_v$ be a restriction as in \cref{restriction}. For $i\in [2kN]$ define $P_i$ by $\frac{1}{1-|z_0(i)|}$. Since $z_0\in [-1/2,1/2]^{2kN}$, we have $P\in [1,2]^{2kN}$. Note that for every $i\in [2kN]$ and $z\in \{-1,1\}^{2kN}$,
\[ \underset{v\sim z_0}{\E}[(\rho_v(z))(i)]=|z_0(i)|sign(z_0(i)) + (1-|z_0(i)|) z(i)= z_0(i) + P_i^{-1}z(i)\]
This implies that $\underset{v\sim z_0}{\E}[\rho_v( z)]=z_0+P^{-1}\cdot z$ for all $z\in \{-1,1\}^{2kN}$. Note that for every $z\in \{-1,1\}^{2kN}$, the multilinear polynomial $H$ and the random variable $\rho_v(z)$ satisfy the hypothesis of \cref{multilinearity} with $a=z_0+P^{-1}\cdot z$. \cref{multilinearity} implies that for all $z\in \{-1,1\}^{2kN}$,
\[\underset{v\sim z_0}{\E}[H(\rho_v( z))] = H(z_0 + P^{-1}\cdot z)\]
Consider the restricted function $H\circ \rho_v$. For every $z\in \{-1,1\}^{2kN}$ and $v\in \{-1,1,0\}^{2kN}$, by definition, $(H\circ \rho_v)(z) =H(\rho_v(z))$. This, along with the previous equality implies that for all $z\in \{-1,1\}^{2kN}$,
\[  \underset{v\sim z_0}{\E}[(H\circ\rho_v)( z)]= H(z_0 + P^{-1}\cdot z)
\]
Note that both the L.H.S. and the R.H.S. of the above equation are multilinear polynomials in $z$ (since we identify $H\circ \rho_v$ with its multilinear extension). Thus, the above equation holds for all $z\in \mathbb{R}^{2kN}$. In particular, for all distributions $D$ over $\mathbb{R}^{2kN}$, it holds that
\begin{equation} \label{expectrestrict}
\underset{z\sim D}{\E} \underset{v\sim z_0}{\E}[(H\circ\rho_v)( z)]= \underset{z\sim D}{\E} [H(z_0 + P^{-1}\cdot z)]
\end{equation}
This implies that $\Delta$ can be expressed as follows.
\begin{align*}\begin{split} \Delta&\triangleq\left|  \underset{z\sim p\cdot \mu_0^{(k)} }{\E}\left[ H(z_0 +  z) \right]-  \underset{z\sim p \cdot \mu_1^{(k)} }{\E}\left[ H(z_0+ z) \right] \right| \\
&=\left|  \underset{z\sim pP\cdot \mu_0^{(k)} }{\E} \left[H(z_0 + P^{-1}\cdot z) \right]-  \underset{z\sim p P\cdot \mu_1^{(k)} }{\E} \left[H(z_0+P^{-1}\cdot z) \right] \right| \\
&= \left| \underset{v\sim z_0}{\E} \left[ \underset{z\sim pP\cdot \mu_0^{(k)} }{\E}[(H\circ\rho_v)(z)]- \underset{z\sim pP\cdot \mu_1^{(k)} }{\E}[ (H\circ\rho_v)( z)]\right] \right|  \quad \text{ \ldots due to \cref{expectrestrict}}\\
&\le \max_{v\sim z_0} \left| \underset{z\sim pP\cdot \mu_0^{(k)} }{\E}[(H\circ\rho_v)(z)]- \underset{z\sim pP\cdot \mu_1^{(k)} }{\E} [(H\circ\rho_v)( z)]\right|  \quad \text{ \ldots Triangle-Inequality}\\
\end{split}\end{align*}
Fix any $v\in \{-1,1,0\}^{2kN}$. We now apply \cref{mainlemma1} on the function $H\circ\rho_v$ with the parameters $2p$ and $pP$. Since $\mathcal{H}$ is closed under restrictions, $H\circ \rho_v\in \mathcal{H}$.  Note that the assumption $p\le \frac{1}{4N}$ and $P\in [1,2]^{2kN}$ implies that $2p\le \frac{1}{2N}$ and $pP\in [-2p,2p]^{2kN}$ and thus, the hypothesis of \cref{mainlemma1} is satisfied. Furthermore, we can bound $L_{2k}(H\circ\rho_v)$ by $L_{2k}(\mathcal{H})$, by definition of the latter. \cref{mainlemma1} implies that
\[\Delta \le   O\left(2^{-2k}\cdot  \frac{L_{2k}(\mathcal{H})  (2p)^{2k}}{N^{k/2}}   +  (2p)^{2(k+1)} N^{(k+1)/2} \right)    \]
This completes the proof of \cref{mainlemma2}.
\end{proof}

\section{Proof of Main Theorem}

For $u,v\in \mathbb{N}^k$, let $\mathbbm{1}_{u=v} \in \{0,1\}$ be the indicator function that is 1 if and only if $u=v$. As mentioned in the preliminaries, we identify sets $S\subseteq[k]$ with their indicator vectors in $\{0,1\}^k$.

\subsection{Proof of \cref{theorem1}}
Let $ \Delta:=   \underset{z\sim \tilde{\mu}_0^{(k)} }{\E}[H(z)]- \underset{z\sim  \tilde{\mu}_1^{(k)} }{\E}[ H( z)] $ be the quantity that we wish to bound in \cref{theorem1}. By the definition of $\tilde{\mu}_0^{(k)}, \tilde{\mu}_1^{(k)}$ as in \cref{roundingharddistributions}, we have
\[ \Delta =\frac{1}{2^{k-1}}   \sum_{S\subseteq [k] } (-1)^{|S|} \underset{z\sim  \tilde{\mathcal{G}}^S\mathcal{U}^{\bar{S}}}{\E} \left[H(z ) \right] \]
Let $S\subseteq [k]$. Note the distribution $\tilde{\mathcal{G}}^S \mathcal{U}^{\bar{S}}$ is obtained by rounding the distribution $\mathcal{G}^S \mathcal{O}^{\bar{S}}$ as in \cref{rounding}. We can thus apply \cref{claim2} to the multilinear polynomial $H(z)$ for the distribution $\mathcal{G}^S\mathcal{O}^{\bar{S}}$ to obtain that $\underset{z\sim  \tilde{\mathcal{G}}^S\mathcal{U}^{\bar{S}}}{\E} \left[H(z ) \right]=\underset{z\sim  \mathcal{G}^S\mathcal{O}^{\bar{S}}}{\E} \left[H(trnc(z) ) \right]. $  This along with the above expression for $\Delta$ implies that
\begin{equation} \label{hequation2}  \Delta =\frac{1}{2^{k-1}}   \sum_{S\subseteq [k] } (-1)^{|S|} \underset{z\sim  \mathcal{G}^S\mathcal{O}^{\bar{S}}}{\E} \left[H(trnc(z) ) \right]    \end{equation}

Let $T=16N^{2k},p=\frac{1}{\sqrt{T}}=\frac{1}{4N^k}$. For each $t\in [T]$ and $j\in [k]$, let $z^{(t)}_j\sim p\cdot \mathcal{G}$ be an independent sample. By convention, $z^{(0)}_j:=0$ for all $j\in[k]$. Let $Z$ refer to the collection $\{z_j^{(t)}\}_{t\in \{0,\ldots,T\},j\in [k]}$ of random variables. For $t\in\{0,\ldots, T\}$ and $j\in [k]$, define $z^{\le (t)}_j := z^{(0)}_j +\ldots +z^{(t)}_j$. Note that the random variable $z_j^{\le(t)}$ has a Gaussian distribution with mean 0 and covariance matrix as $p^2t$ times that of $\mathcal{G}$ for all $j\in [k]$. In particular, $z_j^{\le(T)}$ is distributed according to $\mathcal{G}$ for all $j\in [k]$.

Let $a=(a_1,\ldots,a_k)$ for $a_1,\ldots,a_k\in \{0,\ldots,T\}$. Let $a-1$ denote the vector $(a_1-1,\ldots,a_k-1)$. Let $z^{ (a)}:=(z^{ (a_1)}_1,\ldots,z^{ (a_k)}_k)$ and define $z^{\le (a)}:=(z^{\le (a_1)}_1,\ldots,z^{\le (a_k)}_k)$. Note that $z^{(a)}$ is distributed according to $p\cdot \mathcal{G}^k$ for all $a\in [T]^k$. Also note that $z^{\le(a)}$ is distributed according to $(p\sqrt{a_1}\cdot \mathcal{G}) \times \ldots \times (p\sqrt{a_k}\cdot \mathcal{G})$ for all $a\in \{0,\ldots,T\}^k$. In particular, for every $S\subseteq [k]$, the random variable $z^{\le(T\cdot S)}$ is distributed according to $\mathcal{G}^S \mathcal{O}^{\bar{S}}$. Using this observation in \cref{hequation2}, we have
\begin{equation} \label{hequation} \Delta =\frac{1}{2^{k-1}}   \sum_{S\subseteq[k] } (-1)^{|S|} \underset{Z}{\E} \left[H(trnc(z^{\le (T\cdot S)})) \right] \end{equation}

\begin{claim} \label{telescopic} For $a\in [T]^k$, let $\Delta_a$ be as follows.
\[\Delta_a := \frac{1}{2^{k-1}}  \sum_{S\subseteq [k]}(-1)^{|S|} \underset{Z}{\E}\left[H( trnc({z}^{\le( a-1+S)}) ) \right]     \]
Then, $\underset{a\in [T]^k}{\sum} \Delta_a = \Delta$.
\end{claim}

\begin{proof}[Proof of \cref{telescopic}]
By definition of $\Delta_a$, we have
\[ \label{delta1} 2^{k-1}  \underset{a\in [T]^k}{\sum} \Delta_a = \underset{a\in [T]^k}{\sum} \sum_{S\subseteq [k]} (-1)^{|S|} \underset{Z}{\E}\left[ H(trnc(z^{\le (a-1+S)})) \right]  \]
For every $a\in [T]^k$ and $S\subseteq [k]$, note that $a-1+S\in \{0,\ldots,T\}^k$. Thus, the R.H.S. of the above equation is a linear combination of terms $\underset{Z}{\E}[ H(trnc({z}^{\le(b)}) )]$ for $b\in \{0,\ldots,T\}^k$. That is,
\begin{equation} \label{delta} 2^{k-1} \underset{a\in [T]^k}{\sum} \Delta_a =\underset{b\in \{0,\ldots,T\}^k}{\sum}   \left( \sum_{a\in [T]^k}  \sum_{S\subseteq [k]} \mathbbm{1}_{a-1+S=b}\cdot (-1)^{|S|} \right)   \underset{Z}{\E}[H(trnc({z}^{\le (b)}))]    \end{equation}
We now study the coefficient of $ \underset{Z}{\E}[H(trnc({z}^{\le (b)}))] $ in the R.H.S. of the above expression. Note that $(-1)^{|S|}$ is exactly $\prod_{j=1}^k (1-2S_j)$. For $a\in [T]^k$, let  $a=(a_1,\ldots,a_k)$ for $a_1,\ldots,a_k\in [T]$. Using this notation, the coefficient of $\underset{Z}{\E}[H(trnc(z^{\le(b)}))]$ in \cref{delta} is
\begin{align*}\begin{split}
&\sum_{a\in [T]^k}  \sum_{S\subseteq [k]}   \mathbbm{1}_{a-1+S=b} \cdot (-1)^{|S|}\\
&= \sum_{a\in [T]^k}  \sum_{S\subseteq [k]}   \underset{j\in [k]}{\prod} \left( S_j\cdot \mathbbm{1}_{a_j=b_j} +  (1-S_j) \cdot \mathbbm{1}_{a_j-1=b_j} \right) \cdot (-1)^{|S|} \\
& =  \sum_{a\in[T]^k}  \sum_{S\subseteq [k]} \underset{j\in [k]}{\prod} \left( S_j\cdot \mathbbm{1}_{a_j=b_j} +  (1-S_j) \cdot \mathbbm{1}_{a_j-1=b_j} \right)\cdot  \prod_{j\in [k]} (1-2S_j)  \\
& =  \sum_{a\in[T]^k}  \sum_{S\subseteq [k]}  \underset{j\in [k]}{\prod} \left( S_j (1-2S_j) \cdot \mathbbm{1}_{a_j=b_j} +  (1-S_j)(1-2S_j)  \cdot \mathbbm{1}_{a_j-1=b_j} \right) \\
& =  \sum_{a\in[T]^k}  \sum_{S\subseteq [k]}  \underset{j\in [k]}{\prod} \left( -S_j \cdot \mathbbm{1}_{a_j=b_j} +  (1-S_j) \cdot \mathbbm{1}_{a_j-1=b_j} \right)  \quad\quad \ldots\text{ since } S_j^2=S_j\text{ for all }j\in [k] \\
& =  \sum_{a\in[T]^k}  \underset{j\in [k]}{\prod} \sum_{S_j\in \{0,1\}}\left( -S_j \cdot \mathbbm{1}_{a_j=b_j} +  (1-S_j) \cdot \mathbbm{1}_{a_j-1=b_j} \right)  \\
& =  \sum_{a\in[T]^k}  \underset{j\in [k]}{\prod} \left(  -\mathbbm{1}_{a_j=b_j} +  \mathbbm{1}_{a_j-1=b_j}   \right)  \\
&=   \underset{j\in [k]}{\prod}\sum_{a_j\in[T]}  \left(  -\mathbbm{1}_{a_j=b_j} +  \mathbbm{1}_{a_j-1=b_j} \right)  \\
& =   \underset{j\in [k]}{\prod} \left(  \mathbbm{1}_{0=b_j} -\mathbbm{1}_{T=b_j}  \right)  \\
\end{split}\end{align*}
Note that $\prod_{j\in [k]} \left(  \mathbbm{1}_{0=b_j} -\mathbbm{1}_{T=b_j}  \right) $ is non zero if and only if each coordinate of $b$ is in $\{0,T\}$. For $b\in \{0,T\}^k$, let $B:=\{j \in [k] : b_j=T\}$. Note that $ \underset{j\in [k]}{\prod}\left(  \mathbbm{1}_{0=b_j} -  \mathbbm{1}_{T=b_j} \right)=(-1)^{|B|}$. This, along with the above calculation implies that the coefficient of $\underset{Z}{\E}[H(trnc(z^{\le(b)}))]$ in the R.H.S. of \cref{delta} is precisely $(-1)^{|B|}$. Furthermore, note that $z^{\le(b)}=z^{\le(T\cdot B)}$. We substitute this in \cref{delta} to obtain
\[ 2^{k-1}\underset{a\in [T]^k}{\sum} \Delta_a  =  \underset{Z}{\E}\left[ \underset{B \subseteq [k] }{\sum}(-1)^{|B|}   H(trnc(z^{\le(T\cdot B)}))  \right] \]
This, along with \cref{hequation} completes the proof of \cref{telescopic}.
\end{proof}

Let $a\in [T]^k$. We now show how to bound $\Delta_a$. Let $E_a$ denote the event that $z^{\le (a-1)}\notin [-1/2,1/2]^{2kN}$. We show that $E_a$ is a low probability event. Recall that for $j\in [k],i\in [2N]$, the $(j,i)$-th coordinate of $z^{\le(a-1)}$ is distributed according to $\mathcal{N}(0,p^2 (a_j-1)\epsilon)$, where $p^2a_j\le 1 $ and $\epsilon=1/(60k^2\ln N)$. This implies that for every $i\in [2kN]$,
\[ \p[z^{\le (a-1)}(i)\notin [-1/2,1/2]]\le  \p[ |\mathcal{N}(0,\epsilon)| \ge 1/2 ] \le \exp(-1/(8\epsilon))\le \exp(-7k^2\ln N)\le \frac{1}{N^{7k^2}} \]
Applying a Union bound over coordinates $i\in [2kN]$, we have that for each $a\in [T]^k$,
\begin{equation}\label{equation1} \p[E_a] \triangleq \p[z^{\le (a-1)}\notin [-1/2,1/2]^{2kN}] \le 2kN\cdot \frac{1}{N^{7k^2}} \le \frac{2k}{N^{6k^2}} \end{equation}

\begin{definition}\label{deltaa} For $a\in [T]^k$, let
\[ \Delta_{\neg E_a} : = \frac{1}{2^{k-1}} \sum_{S\subseteq [k]} (-1)^{|S|}  \E_{Z} \left[H(trnc({z}^{\le(a-1+S)})) \mid \neg E_a \right] \]
\[ \Delta_{E_a} : =  \frac{1}{2^{k-1}}\sum_{S\subseteq [k]} (-1)^{|S|} \E_{Z} \left[  H(trnc({z}^{\le (a-1+S)})) \mid   E_a \right] \]
\end{definition}

We bound $\Delta_{\neg E_a}$ as follows. Fix any $z_0:=z^{\le(a-1)}$ such that $E_a$ does not occur. Let $S\subseteq[k]$. Note that by definition, for every fixed $z_0$, the random variable $z^{\le(a-1+S)}$ is distributed according to $z_0+p\cdot \mathcal{G}^S\mathcal{O}^{\bar{S}}$. We now apply \cref{maincorollary2} to the polynomial $H$ with parameters $p$ and $z_0=z^{\le(a-1)}$. The conditions of \cref{maincorollary2} are satisfied, since $z_0\in [-1/2,1/2]^{2kN}$, $p\le \frac{1}{4N^k}\le\frac{1}{4N}$, and for every $S\subseteq[k]$, the random variable $z^{\le(a-1+S)}$ is distributed according to $z_0+ p \cdot \mathcal{G}^S \mathcal{O}^{\bar{S}}$. \cref{maincorollary2} implies that
\begin{equation}\label{equation2} \frac{1}{2^{k-1}}\left|\sum_{S\subseteq [k]} (-1)^{|S|}  \E_{Z} \left[  H(z^{\le(a-1+S)}) \mid \neg E_a \right]  \right|  \le  O\left( 2^{-2k}\cdot \frac{L_{2k}(\mathcal{H}) (2 p)^{2k}}{N^{k/2}}   +  (2p)^{2(k+1)} N^{(k+1)/2} \right)  \end{equation}

Fix any $S\subseteq[k]$. Let $P\in \{0,p\}^{2kN}$ be such that for all $i\in [2N]$ and $j\in [k]$, we have $P_{j,i}= p$ if and only if $j\in S$. Using this notation, observe that for every fixed $z_0$, the random variable $z^{\le(a-1+S)}$ is distributed according to $z_0+P\cdot \mathcal{G}^k$. We now apply \cref{claim5} to the multilinear polynomial $H$ with $z_0={z}^{\le(a-1)}$ and $P$ as defined above. The conditions of this claim are satisfied since $z_0\in [-1/2,1/2]^{2kN}$ (since $E_a$ does not occur), $p\le \frac{1}{4N^k} \le \frac{1}{2}$ and $P\in \left[-p,p\right]^{2kN} \subseteq \left[- \frac{1}{2},\frac{1}{2}\right]^{2kN}$ and $H$ maps $\{-1,1\}^{2kN}$ into $[-1,1]$. Since $z^{\le(a-1+S)}$ is distributed according to $z_0+P\cdot \mathcal{G}^k$, \cref{claim5} implies that for all $S\subseteq [k]$,
\[  \E_{Z} \left[  H(z^{\le(a-1+S)} )-H(trnc({z}^{\le(a-1+S)})) \mid \neg E_a \right] \le O\left(\frac{1}{N^{5k^2}} \right)  \]
This inequality, along with Triangle-Inequality implies that
\begin{align}\begin{split}
\label{equation3}
\frac{1}{2^{k-1}} \left|   \sum_{S\subseteq [k]} (-1)^{|S|}   \E_{Z} \left[\left( H(z^{\le(a-1+S)} )-H(trnc({z}^{\le(a-1+S)}) ) \right) \mid \neg E_a \right]\right| \le O\left(\frac{1}{N^{5k^2}}\right)
\end{split}\end{align}
Combining \cref{equation2} and \cref{equation3} and applying Triangle-Inequality, we have
\begin{align}\begin{split}\label{equation4} |\Delta_{\neg E_a}|&\triangleq\frac{1}{2^{k-1}} \left|  \E_{Z} \left[ \sum_{S\subseteq [k]} (-1)^{|S|} H(trnc({z}^{\le (a-1+S)})) \mid \neg E_a \right]  \right|
\\
&\le  O\left(2^{-2k}\cdot  \frac{L_{2k}(\mathcal{H})  (2p)^{2k}}{N^{k/2}}   +  (2p)^{2(k+1)} N^{(k+1)/2} +\frac{1}{N^{5k^2}} \right)   \end{split}\end{align}
We now bound $\Delta_{E_a}$. For all $a\in[T]^k$ and $S\subseteq[k]$, since $trnc({z}^{\le(a-1+S)})\in [-1,1]^{2kN}$, and $H$ maps $[-1,1]^{2kN}$ to $[-1,1]$, we have $H(trnc({z}^{\le(a-1+S)}))\in [-1,1]$. This, along with the definition of $\Delta_{E_a}$ as in \cref{deltaa} implies that $|\Delta_{E_a}|\le 2$. By the definition of $\Delta_a$ and \cref{deltaa}, we have
\[  |\Delta_a| \le \p[E_a]\cdot |\Delta_{E_a}| + \p[\neg E_a]\cdot |\Delta_{\neg E_a}| \le \p[E_a]\cdot |\Delta_{E_a}| +  |\Delta_{\neg E_a}|\]
Using \cref{equation1}, \cref{equation4}, along with the inequality $|\Delta_{E_a}|\le 2$, we have
\begin{align}\begin{split}
\label{equation6}
|\Delta_a| &\le  O\left( \frac{2k}{N^{6k^2}} +  2^{-2k}\cdot \frac{L_{2k}(\mathcal{H})  (2p)^{2k}}{N^{k/2}}   +  (2p)^{2(k+1)} N^{(k+1)/2} +\frac{1}{N^{5k^2}}  \right) \\
&= O\left(   \frac{L_{2k}(\mathcal{H})  p^{2k}}{N^{k/2}}   +  (2p)^{2(k+1)} N^{(k+1)/2} + \frac{k}{N^{5k^2}}  \right)
\end{split}\end{align}
This establishes a bound on $\Delta_a$. Using \cref{telescopic} and Triangle-Inequality, we have $ |\Delta|\le \underset{a\in [T]^k}{\sum} |\Delta_a|$. Substituting the bound from \cref{equation6} for $\Delta_a$ in this, we have
\begin{align*}\begin{split}
 |\Delta| &\le \underset{a\in [T]^k}{\sum} O\left(   \frac{L_{2k}(\mathcal{H}) p^{2k}}{N^{k/2}}   +  (2p)^{2(k+1)} N^{(k+1)/2} + \frac{k}{N^{5k^2}}  \right) \\
&\le O\left( T^k\cdot \frac{L_{2k}(\mathcal{H})  p^{2k}}{N^{k/2}}   +  T^k\cdot (2p)^{2(k+1)} N^{(k+1)/2} + T^k\cdot \frac{k}{N^{5k^2}}  \right) \\
\end{split}\end{align*}
By our choice of $T=16 N^{2k}$ and $p=\frac{1}{\sqrt{T}}=\frac{1}{4N^k}$, we have the following inequality.
\begin{align*}\begin{split}
 |\Delta| &\le O\left(  \frac{L_{2k}(\mathcal{H}) }{N^{k/2}}   +  16^k N^{2k^2}\cdot \frac{1}{2^{2(k+1)}N^{2k(k+1)}}\cdot  N^{(k+1)/2} + 16^kN^{2k^2}\cdot \frac{k}{N^{5k^2}}  \right) \\
 &\le O\left( \frac{L_{2k}(\mathcal{H}) }{N^{k/2}}   + \frac{4^k}{ N^{2k}} \cdot  N^{(k+1)/2} + \frac{k\cdot 16^k}{N^{3k^2}}  \right) \\
 & \le O\left( \frac{L_{2k}(\mathcal{H}) }{N^{k/2}}   + \frac{4^k}{ N^{\frac{3k-1}{2}}} + \frac{k}{N^{3k^2-k}} \right)
 \end{split}\end{align*}
A small calculation then shows that
\[|\Delta| \le  O\left( \frac{L_{2k}(\mathcal{H}) }{N^{k/2}} \right)  + o\left(\frac{1}{ N^{k/2}} \right) \]
This completes the proof of \cref{theorem1}.

\subsection{Proof of \cref{restatetheorem1}}
\cref{restatetheorem1} essentially follows from the fact that functions in $\mathcal{H}$ are bounded over $\{-1,1\}^{N}$ and the fact that for $i\in \{0,1\}$ the distributions $\sigma_i^{(k)}$ and $\tilde{\mu}_i^{(k)}$ are nearly identical. Let $H\in \mathcal{H}$. Define distributions $\pi_0^{(k)}$ (respectively $\pi_1^{(k)}$) obtained by conditioning $\tilde{\mu}_0^{(k)}$ on $F^{(k)}(z)=-1$ (respectively conditioning $\tilde{\mu}_1^{(k)}$ on $F^{(k)}(z)=+1$).
\cref{concentrationcorollary2} implies for $\delta_0,\delta_1 = O\left(\frac{k}{N^{6k^2}}\right)$, we have $\tilde{\mu}_0^{(k)} = (1-\delta_0) \sigma_0^{(k)} +\delta_0 \pi_0^{(k)}$ and $ \tilde{\mu}_1^{(k)} = (1-\delta_1) \sigma_1^{(k)} + \delta_1 \pi_1^{(k)}$. Thus, for $i\in \{0,1\}$, we have
\[  \underset{ z\sim \tilde{\mu}_i^{(k)} }{\E} [H(z)] = (1-\delta_i)\underset{ z\sim \sigma_i^{(k)}}{\E}[H(z)]  +\delta_i \underset{z\sim \pi_i^{(k)}}{\E}[H(z) ]  \]
Let $\delta=\max(\delta_0,\delta_1)=O\left(\frac{k}{N^{6k^2}}\right)$. Since $H$ maps $\{-1,1\}^{2kN}$ to $[-1,1]$, we may bound $|\E_{ z\sim \sigma_i^{(k)}}[H(z)]| $ and $|\E_{z\sim \pi_i^{(k)}}[H(z) ]|$ by $1$. We subtract the equation for $i=1$ from that for $i=0$ and apply Triangle-inequality to obtain
\[ \left| \underset{ z\sim \tilde{\mu}_0^{(k)} }{\E} [H(z)] - \underset{ z\sim \tilde{\mu}_1^{(k)} }{\E} [H(z)]\right| \ge  \left| \underset{ z\sim \sigma_0^{(k)}}{\E}[H(z)]-\underset{ z\sim \sigma_1^{(k)}}{\E}[H(z)]\right|-3\delta
\]
Rearranging this, we have
\[(*):=\left| \underset{ z\sim \sigma_0^{(k)}}{\E}[H(z)]-\underset{ z\sim \sigma_1^{(k)}}{\E}[H(z)]\right| \le O\left(  \left| \underset{ z\sim \tilde{\mu}_0^{(k)} }{\E} [H(z)] - \underset{ z\sim \tilde{\mu}_1^{(k)} }{\E} [H(z)]\right|  + \delta\right) \]
We use \cref{theorem1} to bound the first term in the R.H.S. Furthermore, we use the fact that $\delta =O\left(\frac{k}{N^{6k^2}}\right)=o\left(\frac{1}{N^{k/2}}\right)$ to obtain that $ (*) \le O\left( \frac{L_{2k}(\mathcal{H})}{N^{k/2}}\right)  +o\left( \frac{1}{N^{k/2}} \right)$. This completes the proof of \cref{restatetheorem1}.

\section{Applications}
\paragraph*{Quantum Upper Bound} The quantum query algorithm for $F^{(k)}$ is derived from \cite{aaronson10,aaronsonambainis}. These papers provide a quantum query algorithm $Q(z)$ which makes one quantum query to the input $z\in \{-1,1\}^{2N}$ and returns a (probabilistic) $b\in \{0,1\}$, with the property that $\p[b=1]=\frac{1+forr(z)}{2}$. Given input $z=(z_1,\ldots,z_k)$ where $z_1,\ldots,z_k\in \{-1,1\}^{2N}$, we are promised that for each $j\in [k]$, either $forr(z_j)\ge \epsilon/2$ or $forr(z_j)\le \epsilon/4$. This implies that for all $j\in [k]$, the probability that $Q(z_j)$ returns 1 is either at least $\frac{1+\epsilon/2}{2}$ or at most $\frac{1+\epsilon/4}{2}$. By repeating the algorithm $O\left(\frac{\log k}{\epsilon^2}\right)$ times and taking the threshold, we can produce an algorithm that for each $j\in [k]$, distinguishes between $F(z_j)=1$ and $F(z_j)=-1$ with probability at least $1-\frac{1}{10k}$. By a Union-bound over $j\in [k]$, with probability at least $9/10$, this algorithm computes $F(z_j)$ for all $j\in [k]$. In particular, it can compute $F^{(k)}(z)=\prod_{j=1}^k F(z_j)$ with probability at least $9/10$. Observe that the number of queries made by this algorithm is $k\times \log k/\epsilon^2 = O\left(k^5\log k\log^2 N \right)$.

It follows that the above algorithm can distinguish the distributions $\sigma_0^{(k)}$ and $\sigma_1^{(k)}$ with at least $9/10$ advantage. A variant of this algorithm can be used to establish the quantum communication protocol in \cref{theorem4}. This step is identical to Theorem 3.3 from \cite{grt}, so we omit it. We now prove the classical lower bounds.

\subsection{Query Complexity Separations}

\begin{proof}[Proof of \cref{theorem4}]

Let $d=o\left(\frac{\sqrt{N}}{\log N}\right)$. Note that $\frac{d\log(kN)}{{\sqrt{N}}}=o(1)$.  \cref{corollary3oftheorem1} implies that every decision tree of depth at most $d$ can distinguish $\sigma_0^{(k)}$ and $\sigma_1^{(k)}$ with advantage at most $ \left( \frac{O(d\log(kN))}{N^{1/2}}
\right)^k\le \exp(-\Omega(k))$.  Note that $\sigma_1^{(k)}$ and $\sigma_0^{(k)}$ are distributions on the \textsc{yes} and \textsc{no} instances of $F^{(k)}$, respectively.
This implies that every randomized decision tree of depth $d=\tilde{o}(\sqrt{N})$ can solve $F^{(k)}$ with at most $\exp(-\Omega(k))$ advantage.
\end{proof}

\begin{proof}[Proof of \cref{corollary3oftheorem1}] Let $\mathcal{H}$ denote the set of Boolean functions on $2kN$ variables that are computed by deterministic decision trees of depth at most $d$. $\mathcal{H}$ is clearly closed under restrictions. We use the following lemma due to \cite{tal} which bounds the level $2k$ mass of $\mathcal{H}$.
\begin{lemma}[\cite{tal}] For all $k\in \mathbb{N}$,  we have $L_{2k}(\mathcal{H}) \le \left(O\left( \sqrt{d\log(kN)}
\right)\right)^{2k}$.
\end{lemma}
The above bound, along with \cref{restatetheorem1} implies that for all $H\in \mathcal{H}$,
\[ \left| \underset{ z\sim \sigma_0^{(k)}}{\E}\left[ H(z) \right]  -  \underset{ z\sim \sigma_1^{(k)}}{\E}\left[ H(z) \right] \right| \le  \left( \frac{O(d\log(kN))}{N^{1/2}}\right)^{k} + o\left(  \frac{1}{N^{k/2}}\right) =\left(   \frac{O(d\log(kN))}{N^{1/2}}\right)^{k}  \]
This completes the proof of \cref{corollary3oftheorem1}.
\end{proof}

\subsection{Circuit Complexity Separations}

\begin{proof}[Proof of \cref{theorem2}] Let $C$ be an AC0 circuit of depth $d$ and size $s=o\left( \exp\left( N^{\frac{1}{4(d-1)}}\right) \right)$. Note that $O\left( \log^{2d-2}(s) \right)=o(\sqrt{N})$. This, along with \cref{corollary1oftheorem1} implies that
\[ \left| \underset{z\sim \sigma_0^{(k)}}{\E}[C(z)] -  \underset{ z\sim \sigma_1^{(k)}}{\E}[C(z)]  \right| \le \left(\frac{O\left( \log^{2d-2}(s) \right)}{N^{1/2}} \right)^k \le \exp(-\Omega(k)) \]
Thus, we have produced distributions on \textsc{yes} and \textsc{no} instances of $F^{(k)}$ such that every depth $d$ AC0 circuit of size $o\left( \exp\left( N^{\frac{1}{4(d-1)}}\right) \right)$ can distinguish them with at most $\exp(-\Omega(k))$ advantage. This completes the proof of \cref{theorem2}.\end{proof}

\begin{proof}[Proof of \cref{corollary1oftheorem1}]
Let $\mathcal{H}$ denote the set of Boolean functions that are computed by AC0 circuits of depth at most $d$ and size at most $s$. Note that $\mathcal{H}$ is clearly closed under restrictions. We use the following lemma due to \cite{tal} which bounds the level $2k$ mass of $\mathcal{H}$.
\begin{lemma}[\cite{tal}] For all $k\in \mathbb{N}$,  we have $L_{2k}(\mathcal{H}) \le \left( O\left( \log^{d-1}
(s)\right)\right)^{2k}$.
\end{lemma}
The above bound, along with \cref{theorem1} implies that for all $H\in \mathcal{H}$,
\[ \left| \underset{ z\sim \tilde{\mu}_0^{(k)}}{\E}\left[ H(z) \right]  -  \underset{ z\sim \tilde{\mu}_0^{(k)}}{\E}\left[ H(z) \right] \right| \le  \frac{\left(O\left( \log^{2d-2}
(s)\right)\right)^k}{N^{k/2}} + o\left(  \frac{1}{N^{k/2}}\right) = \left( \frac{O\left(\log^{2d-2}
(s)\right)}{N^{1/2}}\right)^k  \]
This completes the proof of \cref{corollary1oftheorem1}.
\end{proof}

\subsection{Applications to Communication Complexity Separations}

\begin{proof}[Proof of \cref{theorem3}]
Let $c=\tilde{o}(N^{1/4})$. Note that for $k=\tilde{o}(N^{1/4})$, we have $\frac{(c+8k)^2}{\sqrt{N}}=o(1).$ For $i\in \{0,1\}$, let $\pi_i^{(k)}$ denote the distribution of $(x,x\cdot z )$ where $x\sim U_{2kN}$ and $z\sim\sigma_i^{(k)}$. Note that $\pi_0^{(k)}$ and $\pi_1^{(k)}$ are distributions on the \textsc{yes} and \textsc{no} instances of $F^{(k)}\circ \textsc{xor}$, respectively. \cref{corollary2oftheorem1} implies that every deterministic protocol of cost at most $c$ for $F^{(k)}\circ \textsc{xor}$ can distinguish $\pi_0^{(k)}$ and $\pi_1^{(k)}$ with at most $O\left( \frac{(c+8k)^{2k}}{N^{k/2}}\right)\le \exp(-\Omega(k))$ advantage. This implies that no randomized protocol of cost $o(N^{1/4})$ solves $F^{(k)}(\oplus)$ with more than $\exp(-\Omega(k))$ advantage. This completes the proof of \cref{theorem3}.\end{proof}

To prove \cref{corollary2oftheorem1}, the idea is to apply \cref{restatetheorem1} on the function family defined by $\E_{x\sim U_{2kN}}C(x,x\cdot z)$, where $C$ is a small cost protocol. However, to prove a suitable upper bound on the level $2k$ mass, we require that each rectangle in the protocol is small. To handle this, we define an {\it extended} protocol $ext^l(C)$, in which the players reveal $l$ additional junk bits and then proceed with the original protocol $C$. This modification is only a technicality and the rest of the arguments are similar to the ones in \cite{grt}.

\begin{definition}
\label{extension}
Let $C:\{-1,1\}^M\times \{-1,1\}^M\rightarrow \{-1,1\}$ be any deterministic protocol and $l\in \mathbb{N}$. An extension $ext^l(C):\{-1,1\}^{M+ l}\times \{-1,1\}^{M+l}\rightarrow \{-1,1\}$ is a protocol in which Alice and Bob declare the last $l$ bits of their inputs and then follow $C$ on the first $M$ bits of their inputs.
\end{definition}

\begin{definition}
\label{xorprotocol} For any protocol $C:\{-1,1\}^M\times \{-1,1\}^M\rightarrow \{-1,1\}$, let $H_C:\{-1,1\}^M\rightarrow \mathbb{R}$ be defined at every $z\in \{-1,1\}^M$ by $H_C(z):=\underset{x\sim U_{M}}{\E}[C(x,x\cdot z)]$. For any distribution $\mathcal{C}$ over protocols $C:\{-1,1\}^M\times \{-1,1\}^M\rightarrow \{-1,1\}$, let $H_{\mathcal{C}}$ be defined at every $z\in \{-1,1\}^M$ by $H_{\mathcal{C}}(z):=\underset{C\sim \mathcal{C}}{\E}[H_C(z)]$.
\end{definition}

\begin{lemma} \label{closedunderrestrictions}  Let $l,M\in \mathbb{N}$. Let $\mathcal{H}$ be the family of functions $H$ obtained as follows. Let $\mathcal{C}$ be an arbitrary distribution over deterministic protocols $C:\{-1,1\}^{M}\times \{-1,1\}^{M}\rightarrow \{-1,1\}$ of cost at most $c$. Let $H_{ext^l(\mathcal{C})}$ be as in \cref{xorprotocol}, and \cref{extension} and let $H:\{-1,1\}^{M}\rightarrow\mathbb{R}$ be defined at every $z\in \{-1,1\}^{M}$ by $ H(z):=\underset{z'\sim U_l}{\E} [H_{ext^l(\mathcal{C})}(z,z')]$. Then, $\mathcal{H}$ is closed under restrictions.
\end{lemma}

The proof of this is a simple unravelling of definitions and is deferred to the appendix.

\begin{lemma} \label{weightbound3} Let $l= \lceil 2k\log e\rceil$. Let $\mathcal{H}$ be the family as in \cref{closedunderrestrictions}. Then, $L_{2k}(\mathcal{H})\le O\left( \left(\frac{e}{k}\right)^{2k}\cdot (c+2l)^{2k} \right)$.
\end{lemma}

The proof of this is similar to that of Claim 1 in \cite{grt} and is deferred to the appendix.

\begin{proof}[Proof of \cref{corollary2oftheorem1}]

Let $l=\lceil 2k\log e\rceil$. Let $\mathcal{H}$ be the family of functions as in \cref{closedunderrestrictions}. \cref{closedunderrestrictions} implies that the family $\mathcal{H}$ is closed under restrictions. We now apply \cref{restatetheorem1} to $\mathcal{H}$ to obtain that for all $H\in \mathcal{H}$,
\begin{equation*} \label{maincorollary2eqn1} \left| \underset{ z\sim \sigma_0^{(k)}}{\E}\left[ H(z) \right]  -  \underset{ z\sim \sigma_1^{(k)}}{\E}\left[ H(z) \right] \right| \le O\left( \frac{L_{2k}(\mathcal{H})}{N^{k/2}}\right) + o\left(  \frac{1}{N^{k/2}}\right)  \end{equation*}
We use \cref{weightbound3} which upper bounds $L_{2k}(\mathcal{H})$. This, along with the previous inequality and the fact that $l=\lceil 2k\log e\rceil$ implies that
\begin{equation} \label{maincorollary2eqn2} \left| \underset{ z\sim \sigma_0^{(k)}}{\E}\left[ H(z) \right]  -  \underset{ z\sim \sigma_1^{(k)}}{\E}\left[ H(z) \right] \right| \le O\left( \frac{e^{2k}(c+2l)^{2k}}{k^{2k}N^{k/2}}\right) + o\left(  \frac{1}{N^{k/2}}\right)= O\left( \frac{(c+8k)^{2k}}{N^{k/2}}\right) + o\left(  \frac{1}{N^{k/2}}\right)  \end{equation}
Let $C$ refer to the given protocol of cost at most $c$. Let $H:\{-1,1\}^{2kN}\rightarrow[-1,1]$ be defined at $z\in \{-1,1\}^{2kN}$ by $H(z)=\E_{z'\sim U_l} [H_{ext^l(C)}(z,z')]$. By \cref{extension}, for all $x,z\in \{-1,1\}^{2kN}, x',z' \in\{-1,1\}^l$, we have that $C(x,x\cdot z) = ext^l(C)((x,x'),(x\cdot z, x'\cdot z'))$.  This implies that for all $z\in \{-1,1\}^{2kN}$, we have
\begin{align*}\begin{split}
H(z) &\triangleq \underset{z'\sim U_l}{\E} [H_{ext^l(C)}(z, z')] \triangleq \underset{\substack{x\sim U_{2kN}\\x',z'\sim U_l}}{\E} [ext^l(C)((x,x'),(x\cdot z, x'\cdot z'))] \quad\ldots\text{due to \cref{xorprotocol}} \\
& = \underset{x\sim U_{2kN}}{\E}[ C(x,x\cdot z) ]\quad\ldots\text{due to \cref{extension}}\\\end{split}\end{align*}
This, along with \cref{maincorollary2eqn2} implies that
\[ \left|  \underset{\substack{x\sim U_{2kN}\\ z\sim \sigma_0^{(k)}}}{\E}\left[ C(x,x\cdot z) \right]  - \underset{\substack{x\sim U_{2kN}\\ z\sim \sigma_1^{(k)}}}{\E}\left[ C(x,x\cdot z) \right]   \right|\le O\left( \frac{(c+8k)^{2k}}{N^{k/2}}\right)+ o\left(  \frac{1}{N^{k/2}}\right) = O\left( \frac{(c+8k)^{2k}}{N^{k/2}} \right)\]
This completes the proof of \cref{corollary2oftheorem1}.
\end{proof}

\section*{Acknowledgement}
We would like to thank Avishay Tal for very helpful conversations.

\bibliographystyle{alpha}

\begin{thebibliography}{1}

\bibitem[A10]{aaronson10}
Scott Aaronson:
BQP and the Polynomial Hierarchy. STOC 2010: 141-150

\bibitem[AA15]{aaronsonambainis}
Scott Aaronson and Andris Ambainis:
Forrelation: A Problem That Optimally Separates Quantum from Classical Computing. STOC 2015. 307-316

\bibitem[BCW98]{buhrman}
Harry Buhrman, Richard Cleve, Avi Wigderson:
Quantum vs. Classical Communication and Computation. STOC 1998: 63-68

\bibitem[BJK04]{Bar-YossefJK04}
Ziv Bar-Yossef, T. S. Jayram, Iordanis Kerenidis:
Exponential Separation of Quantum and Classical One-Way Communication Complexity. SIAM J. Comput. 38(1): 366-384 (2008)

\bibitem[CFK+19]{bppip}
Arkadev Chattopadhyay, Yuval Filmus, Sajin Koroth, Or Meir, Toniann Pitassi:
Query-To-Communication Lifting for BPP Using Inner Product. ICALP 2019: 35:1-35:15

\bibitem[CHHL18]{chhl}
Eshan Chattopadhyay, Pooya Hatami, Kaave Hosseini, Shachar Lovett:
Pseudorandom Generators from Polarizing Random Walks. CCC 2018: 1:1-1:21

\bibitem[CHLT19]{chlt}
Eshan Chattopadhyay, Pooya Hatami, Shachar Lovett, Avishay Tal:
Pseudorandom Generators from the Second Fourier Level and Applications to AC0 with Parity Gates. ITCS 2019: 22:1-22:15

\bibitem[D12]{drucker}
Andrew Drucker: Improved Direct Product Theorems for Randomized Query Complexity. Computational Complexity 21(2): 197-244 (2012)

\bibitem[G16]{gavinsky}	
Dmitry Gavinsky:
Entangled Simultaneity versus Classical Interactivity in Communication Complexity. STOC 2016: 877-884

\bibitem[GKK+08]{gavinskyetal}	
Dmitry Gavinsky, Julia Kempe, Iordanis Kerenidis, Ran Raz, Ronald de Wolf:
Exponential Separation for One-Way Quantum Communication Complexity, with Applications to Cryptography. SIAM J. Comput. 38(5): 1695-1708 (2008)

\bibitem[GRT19]{grt}
Uma Girish, Ran Raz, Avishay Tal:
Quantum versus Randomized Communication Complexity, with Efficient Players. CoRR abs/1911.02218 (2019)

\bibitem[KR11]{klartagregev}
Oded Regev, Bo\`az Klartag:
Quantum One-Way Communication can be Exponentially Stronger than Classical Communication. STOC 2011: 31-40

\bibitem[O'D14]{odonnell}	
Ryan O'Donnell:
Analysis of Boolean Functions. Cambridge University Press 2014, ISBN 978-1-10-703832-5, pp. I-XX, 1-423

\bibitem[R99]{raz}	
Ran Raz:
Exponential Separation of Quantum and Classical Communication Complexity. STOC 1999: 358-367

\bibitem[R95]{razxor}
Ran Raz: Fourier Analysis for Probabilistic Communication Complexity. Comput. Complex. 5(3/4): 205-221 (1995)


\bibitem[RT19]{raztal}
Ran Raz and Avishay Tal:
Oracle separation of {BQP} and {PH}. STOC 2019: 13-23

\bibitem[Tal17]{talac}
Avishay Tal:
Tight Bounds on the Fourier Spectrum of AC0. Computational Complexity Conference 2017: 15:1-15:31

\bibitem[Tal19]{tal}
Avishay Tal:
Towards Optimal Separations between Quantum and Randomized Query Complexities. CoRR abs/1912.12561 (2019)

\bibitem[UCB]{gaussian}
Example 2.1 from \url{https://www.stat.berkeley.edu/~mjwain/stat210b/Chap2_TailBounds_Jan22_2015.pdf}

\bibitem[UCB]{chi}
Example 2.5 from \url{https://www.stat.berkeley.edu/~mjwain/stat210b/Chap2_TailBounds_Jan22_2015.pdf}


\end{thebibliography}

\appendix
\section{Output of $F^{(k)}$ on Distributions $\tilde{\mu}_0^{(k)}$ and $\tilde{\mu}_1^{(k)}$}

We use the following claims to prove \cref{concentrationcorollary2}.

\begin{claim} \label{concentrationlemma1} Let $z\sim \mathcal{G}$, where $\mathcal{G}$ is the distribution in \cref{gaussian}. Then, $ \underset{z\sim\mathcal{G}}{\p}[ forr(z)\le 3\epsilon/4  ] \le e^{-\Omega(N)}$.
\end{claim}
\begin{claim} \label{concentrationlemma2} Let $z_0\in [-1/2,1/2]^{2N}$ and $z\sim\tilde{z_0}$ be the random variable obtained by rounding $z_0$ as in \cref{rounding}. Then, $ \p[ \left| forr(z)-forr(z_0)\right|\ge \epsilon/4  ] \le e^{-\Omega(N^{1/4})}$.
\end{claim}

\begin{corollary}\label{concentrationcorollary1} Let $\mathcal{U}$ be the uniform distribution on $\{-1,1\}^{2N}$ and $\tilde{\mathcal{G}}$ be the distribution on $\{-1,1\}^{2N}$ as in \cref{roundingharddistributions}. Then,
\[ \underset{z\sim \mathcal{U}}{\p}[ forr(z)\le \epsilon/4 ] \ge 1 - e^{-\Omega(N^{1/4})} \quad\text{ and }\quad \underset{z\sim \tilde{\mathcal{G}} }{\p}[ forr(z)\ge \epsilon/2 ] \ge 1-  O\left( \frac{1}{N^{6k^2}} \right) \]
\end{corollary}

\begin{proof}[Proof of \cref{concentrationcorollary2} from \cref{concentrationcorollary1}]
This follows from a simple Union-bound. Let $S\subseteq [k]$. Let $z\sim \tilde{\mathcal{G}}^S\mathcal{U}^{\bar{S}}$ and $z=(z_1,\ldots,z_k)$ for $z_1,\ldots,z_k\in \{-1,1\}^{2N}$. For $j\in S$, we have $z_j\sim \mathcal{G}$ and consequently, \cref{concentrationcorollary1} implies that with at least $1-  O\left( \frac{1}{N^{6k^2}} \right) $ probability, $F(z_j)=-1$. For $j\notin S$, we have $z_j\sim \mathcal{U}$ and consequently, \cref{concentrationcorollary1} implies that with at least $1- e^{-\Omega(N^{1/4})}\ge1- O\left( \frac{1}{N^{6k^2}} \right)$ probability\footnote{Here we use the fact that $k=o(N^{1/50})$.}, $F(z_j)=1$. A Union-bound over $j\in [k]$ implies that with probability at least $1-O\left( \frac{k}{N^{6k^2}} \right)$, we have that all these events occur, that is, $z$ is in the support of $F^{(k)}$ and $F^{(k)}(z)\triangleq\prod_{j=1}^k F(z_j)=(-1)^{|S|}$. Since $\mu_0^{(k)}$ (respectively $\mu_1^{(k)}$) is a mixture of distributions $\tilde{\mathcal{G}}^S\mathcal{U}^{\bar{S}}$ where $|S|$ is even (respectively $|S|$ is odd), it follows that with probability at least $1-O\left( \frac{k}{N^{6k^2}} \right)$, $F^{(k)}(z)=1$ (respectively $F^{(k)}(z)=-1$).
\end{proof}

\begin{proof}[Proof of \cref{concentrationcorollary1} from \cref{concentrationlemma1} and \cref{concentrationlemma2}]
We set $z_0$ to be the zero vector in $\mathbb{R}^{2N}$ and apply \cref{concentrationlemma2}. Since the distribution obtained by rounding $z_0$ is $U_{2N}$ and $forr(z_0)=0$, we have
\[ \underset{z\sim U_{2N} }{\p}\left[ forr(z)\ge \frac{\epsilon}{4}\right] \triangleq\underset{z\sim \tilde{z_0} }{\p}\left[ forr(z)\ge \frac{\epsilon}{4}\right] \le e^{-\Omega(N^{1/4})} \]
This proves the first part of \cref{concentrationcorollary1}. To prove the second part, let $z_0\sim \mathcal{G}$. Let $E$ denote the event that $z_0\notin[-1/2,1/2]^{2N}$. We first show that $E$ is a low probability event. Recall that each coordinate of $z_0$ is distributed as $\mathcal{N}(0,\epsilon)$ where $\epsilon=1/(60k^2\ln N)$. This, along with a Union bound over coordinates $i\in [2N]$ implies that
\begin{align}\begin{split}\label{bound1}
\p[E] &\le 2N\cdot \p[z_0(i)\notin [-1/2,1/2]] \le 2N\cdot \p[ |\mathcal{N}(0,\epsilon)| \ge 1/2 ]\\
& \le2N \exp(-1/(8\epsilon))\le 2N\cdot \exp(-7k^2\ln N)= \frac{2N}{N^{7k^2}} \end{split}\end{align}
Let $z\sim \tilde{z}_0$ be obtained by rounding $z_0$ as in \cref{rounding}. If $forr(z)\le \epsilon/2$, then we must either have $forr(z_0)\le 3\epsilon/4$ or $|forr(z)-forr(z_0)|\ge \epsilon/4$. For the latter event, we split it into cases conditioned on whether $E$ occurs or not. A Union bound implies that
\begin{align}\begin{split} \label{unionbound}
& \underset{\substack{z_0\sim \mathcal{G}\\z\sim\tilde{z}_0}}{\p} [  forr(z) \le \epsilon/2]   \le \underset{z_0\sim \mathcal{G}}{\p}[forr(z_0)\le 3\epsilon/4] +  \underset{\substack{z_0\sim\mathcal{G}\\z\sim \tilde{z}_0}}{\p} [ | forr(z) - forr(z_0)| \ge \epsilon/4] \\
&\le  \underset{z_0\sim \mathcal{G}}{\p}[forr(z_0)\le 3\epsilon/4] + \p[E]+ \underset{\substack{z_0\sim\mathcal{G}\\z\sim \tilde{z}_0}}{\p} [| forr(z) - forr(z_0)| \ge \epsilon/4 \mid \neg E]
 \end{split}\end{align}
\cref{concentrationlemma1} implies that with all but $e^{-\Omega(N)}$ probability, for $z_0\sim \mathcal{G}$, we have $forr(z_0)> 3\epsilon/4$. Thus, the first term in the R.H.S. of \cref{unionbound} can be upper bounded by $e^{-\Omega(N)}.$ The second term can be bounded by $\frac{2N}{N^{7k^2}}$ due to \cref{bound1}. For the third term, note that whenever $E$ does not occur, we can apply \cref{concentrationlemma2} to obtain that
\[ \underset{z\sim \tilde{z}_0}{\p} [ |forr(z) - forr(z_0)| \ge \epsilon/4  \mid \neg E,z_0] \le e^{-\Omega(N^{1/4})} \]
These observations along with \cref{unionbound} imply that
\[  \underset{z\sim\tilde{\mathcal{G}}}{\p} [  forr(z) \le \epsilon/2] \triangleq \underset{\substack{z_0\sim \mathcal{G}\\z\sim\tilde{z}_0}}{\p} [  forr(z) \le \epsilon/2]  \le e^{-\Omega(N)} + \frac{2N}{N^{7k^2}} + e^{-\Omega(N^{1/4})} =O\left( \frac{1}{N^{6k^2}} \right) \]
\end{proof}

\begin{proof}[Proof of \cref{concentrationlemma1}]
This follows from a simple concentration inequality for Chi-Squared random variables. Note that a random sample $z\sim \mathcal{G}$ is equivalent to a sample $z=(x,y)$, where $x\sim \mathcal{N}(0,\epsilon \mathbb{I}_N)$ and $y=H_Nx$. This implies that $ forr(z) =\frac{1}{N} \left< x, H_N y\right> = \frac{1}{N}\langle x , H_N^2 x\rangle =\frac{1}{N}\cdot \|x\|^2 $. The random variable $\|x\|^2$ has a Chi-Squared distribution, defined by the sum of squares of $N$ random variables, each of which is distributed according to $\mathcal{N}(0,\epsilon)$. Using the concentration inequality for the Chi-Squared distribution from the preliminaries, we have that for all $t\in (0,1)$,
\[ \p\left[ \left| \frac{1}{N} \sum_{i=1}^N x_i^2 - \epsilon \right| \ge t\epsilon \right] \le \exp(-\Omega(Nt^2)) \]
Substituting $t=1/4$, we obtain  $\p[|forr(z)-\epsilon|\ge \epsilon/4]= \p\left[ \left| \frac{1}{N} \sum_{i=1}^N x_i^2 - \epsilon \right| \ge \frac{\epsilon}{4} \right] \le e^{-\Omega(N)}$. This implies the desired conclusion in \cref{concentrationlemma1}.
\end{proof}

\begin{proof}[Proof of \cref{concentrationlemma2}]
We make use of the following concentration inequality. It appears as Theorem 10.24 in Ryan Odonnell's book on Boolean functions~\cite{odonnell} as an application of the general hypercontractivity theorem on product spaces. We state it in the context of biased product distributions on the Boolean hypercube.
\begin{lemma} \label{hypercontractivity1} Let $\pi_1,\ldots,\pi_M$ be probability distributions on $\{-1,1\}$ such that for every $i\in [M]$, every outcome in $\pi_i$ has probability at least $\lambda$. Let $\Omega=\{-1,1\}^M$ and $\pi= \pi_1\times\ldots\times \pi_M$. Let $f:\Omega\rightarrow \mathbb{R}$ be a Boolean function of total degree at most $d$ and let $\|f\|_2:=\sqrt{\E_{x\sim \pi}[f(x)^2]}$ denote the $l_2$ norm of $f$. Then, for any $t\ge \sqrt{2e/\lambda}^d$, we have $\underset{x\sim \pi}{\p}[ |f(x)|\ge t\|f\|_2 ] \le \lambda^d \exp\left(-\frac{d}{2e}\cdot \lambda t^{2/d}\right)$.
\end{lemma}
Note that the distribution $\tilde{z}_0$ on $\{-1,1\}^{2N}$ satisfies the hypothesis in \cref{hypercontractivity1} with $\lambda=\frac{1}{4}$ because of the assumption that $z_0\in [1/2,1/2]^{2N}$. \cref{concentrationlemma2} essentially follows by considering the degree-2 Boolean function $forr(z)-forr(z_0)$, bounding its $l_2$ norm and applying \cref{hypercontractivity1}. However, to simplify the calculation we instead consider $f:\mathbb{R}^{2N}\rightarrow \mathbb{R}$ defined by $ f(z):= forr(z-z_0) \triangleq N^{-1}\cdot \left< x-x_0 , H_N( y-y_0)\right>$ where $z=(x,y)$ for $x,y\in\mathbb{R}^N$ and $z_0=(x_0,y_0)$ for $x_0,y_0\in [-1/2,1/2]^{N}$. Note that we have the identity $f(z)\triangleq forr(z-z_0)=forr(z) -forr(x_0,y) -forr(x,y_0)+forr(z_0)$. We now show that when $z\sim \tilde{z}_0$, the random variables $f(z), forr(x,y_0)$ and $forr(x_0,y)$ are concentrated around their mean. From the above identity, it will follow that $forr(z)$ is also concentrated around its mean. We first show a concentration inequality for $f$. Since each coordinate of $(x,y)$ is sampled independently so that $\E[(x,y)]=(x_0,y_0)$, we have
\begin{align*}\begin{split}
\E[f^2]&\triangleq N^{-3}\cdot \E \left[ \left(\sum_{i,j\in [N]}(x(i)-x_0(i))(y(j)-y_0(j))(-1)^{\left<i,j\right>_2}\right)^2\right]\\
&=N^{-3}\cdot\sum_{i,j\in [N]} \E\left[ (x(i)-x_0(i))^2(y(j)-y_0(j))^2\right] \quad\ldots\text{ since the cross terms are 0.}\\
&\le N^{-3}\cdot 16N^2  \quad\quad\ldots \text{since } x,x_0,y,y_0\in [-1,1]^{2N}.
\end{split}\end{align*}
Thus, $\|f\|_2\le \frac{4}{\sqrt{N}}$. Note that $f$ is of degree $2$. We now apply \cref{hypercontractivity1} to the function $f$ for the distribution $\tilde{z}_0$. Let $t$ be a parameter. Since $\lambda=\frac{1}{4}$ and $d=2$, we have $\sqrt{2e/\lambda}^d=O(1)$ and $\lambda^d \exp\left(-\frac{d}{2e}\cdot \lambda t^{2/d}\right)= \exp(-\Omega(t))$. \cref{hypercontractivity1}, along with the above calculation implies that for all $t\ge O(1)$, we have $ \p_{z\sim\tilde{z}_0} \left[ |f(z)| \ge \frac{t}{N^{1/2}} \right] \le   \p_{z\sim\tilde{z}_0} \left[ |f(z)| \ge \Omega(t)\cdot \|f\|_2 \right] \le \exp(-\Omega(t))$. We now set $t=\frac{N^{1/2}\epsilon}{12}=\frac{N^{1/2}}{720k^2\ln N}$. This is larger than $N^{1/4}$ for sufficiently large $N$ and $k=o(N^{1/50})$. This implies that
\begin{equation}\label{concentrationlemma2eqn1}  \p_{z\sim\tilde{z}_0} \left[ |forr(z-z_0)|\ge \frac{\epsilon}{12} \right] \triangleq \p_{z\sim\tilde{z}_0} \left[ |f(z)| \ge \frac{\epsilon}{12} \right]\le \exp(-\Omega(N^{1/4}))  \end{equation}
We now show a similar concentration inequality for $forr(x,y_0)$. Let $g:\mathbb{R}^N\rightarrow\mathbb{R}$ be defined at $x\in \mathbb{R}^N$ by $g(x):=forr(x,y_0)-forr(z_0)\triangleq N^{-1}\cdot \left<x-x_0,H_Ny_0\right> $. Since each coordinate of $x$ is sampled independently so that $\E[x]=x_0$, we have
\begin{align*}\begin{split}
N^2\cdot \E[g^2]&\triangleq\E\left[ \left(\sum_{i\in [N]}(x(i)-x_0(i))(H_Ny_0)(i)\right)^2 \right]  \\
&= \sum_{i\in [N]} \E\left[ (x(i)-x_0(i))^2(H_Ny_0)(i)^2 \right] \quad\ldots \text{since the cross terms are 0.}\\
&\le \sum_{i\in [N]} 4\cdot  (H_Ny_0)(i)^2  \quad\quad\ldots \text{since }x,x_0\in [-1,1]^{N}. \\
&= 4\|H_N y_0\|_2^2 = 4\|y_0\|_2^2 \le 4N
\end{split}\end{align*}
Thus, $\|g\|_2\le \frac{2}{\sqrt{N}}$. We now apply \cref{hypercontractivity1} to the degree-1 polynomial $g$ for the distribution $\tilde{x}_0$ on $\{-1,1\}^{N}$. Let $t$ be a parameter. Since $\lambda=\frac{1}{4}$ and $d=1$, we have $\lambda^d \exp\left(-\frac{d}{2e}\cdot \lambda t^{2/d}\right)= \exp(-\Omega(t^2))$ and $\sqrt{2e/\lambda}^d=O(1)$. \cref{hypercontractivity1}, along with the above calculation implies that for all $t\ge O(1)$, we have
$ \p_{x\sim\tilde{x}_0} \left[ |g(x)| \ge \frac{t}{N^{1/2}} \right] \le   \p_{x\sim \tilde{x}_0} \left[ |g(x)| \ge \Omega(t)\cdot \|g\|_2 \right] \le \exp(-\Omega(t^2))$. We now set $t=\frac{N^{1/2}\epsilon}{12}=\frac{N^{1/2}}{720k^2\ln N}$. This is larger than $N^{1/4}$ for sufficiently large $N$ and $k=o(N^{1/50})$. This implies that
\begin{equation} \label{concentrationlemma2eqn2} \p_{x\sim \tilde{x}_0}\left[|forr(x,y_0)-forr(z_0)| \ge \epsilon/12 \right]\triangleq \p_{x\sim\tilde{x}_0} \left[ |g(x)| \ge \frac{\epsilon}{12} \right] \le \exp(-\Omega(N^{1/2})) \end{equation}
An identical calculation implies that
\begin{equation} \label{concentrationlemma2eqn3} \p_{y\sim \tilde{y}_0}\left[|forr(x_0,y)-forr(z_0)| \ge \epsilon/12 \right]\le \exp(-\Omega(N^{1/2})) \end{equation}
Recall that we have the identity $forr(z)= forr(z-z_0)+forr(x_0,y) +forr(x,y_0)-forr(z_0)$. Suppose $|forr(z)-forr(z_0)|\ge \epsilon/4$, then either $|forr(x,y_0)-forr(z_0)|\ge \epsilon/12$, or $|forr(x_0,y)-forr(z_0)|\ge\epsilon/12$ or  $|forr(z-z_0)|\ge \epsilon/12$. This, along with \cref{concentrationlemma2eqn1}, \cref{concentrationlemma2eqn2}, \cref{concentrationlemma2eqn3} and a Union-Bound implies that
\[ \p\left[ |forr(z)-forr(z_0)|\ge \epsilon/4 \right] \le 2\cdot e^{-\Omega(N^{1/2})}+e^{-\Omega(N^{1/4})} \le e^{-\Omega(N^{1/4})}\]
\end{proof}

\section{Closure Under Restrictions}

\begin{proof}[Proof of \cref{closedunderrestrictions}]
Let $L=M+l$. Let $H\in \mathcal{H}$ be defined by a distribution $\mathcal{C}$ over deterministic protocols $C:\{-1,1\}^M\times \{-1,1\}^M\rightarrow \{-1,1\}$ of cost at most $c$. Let $v\in \{-1,1,0\}^{M}$ and $\rho_v$ be a restriction as in \cref{restriction}. Let $V:=\{ j : v(j)\in \{-1,1\}\}$. Define a distribution $\mathcal{C}_{v}$ over protocols $C_v:\{-1,1\}^{M}\times \{-1,1\}^{M} \rightarrow \{-1,1\}$ as follows.
\begin{enumerate}
\item Sample $C\sim \mathcal{C}$.
\item For each $j\in V$, independently sample $a_j$ uniformly at random from $\{-1,1\}$.
\item For each $j\in V$, Alice overwrites the $j$-th bit of her input with $a_j$ and Bob overwrites the $j$-th bit of his input with $a_j\cdot v_j$.
\item Alice and Bob execute the protocol $C$ on their restricted inputs.
\end{enumerate}
\begin{claim} \label{closedunderrestrictionslemma} For all $z\in \{-1,1\}^{M}$, $v\in \{-1,1,0\}^{M}$, we have $\underset{z'\sim U_{l}}{\E} [H_{ext^l(\mathcal{C}_v)}(z,z')]= H(\rho_v(z))$.
\end{claim}
Note that $\mathcal{C}_v$ is a distribution over deterministic protocols of cost at most $c$. Thus, by definition of $\mathcal{H}$, the function that maps $z$ to $\underset{z'\sim U_{l}}{\E}[H_{ext^l(\mathcal{C}_v)}(z,z')]$ is in $\mathcal{H}$. This observation, along with \cref{closedunderrestrictionslemma} establishes that the restricted function $H(\rho_v(z))$ of $z$ is also in $\mathcal{H}$. It thus suffices to prove \cref{closedunderrestrictionslemma}. \end{proof}

\begin{proof}[Proof of \cref{closedunderrestrictionslemma}]  This proof is by unravelling definitions. Let $z\in \{-1,1\}^M$ and $v\in\{-1,1,0\}^M$. Note that for all $x,'z'\in \{-1,1\}^l$ and $x\in \{-1,1\}^M$, \cref{extension} implies that $C(x,x\cdot z)= ext^l(C)((x,x'),(x\cdot z,x'\cdot z')$. In particular, for all $x\in \{-1,1\}^M$, we have
\begin{equation}\label{usefulrestriction}
C(x,x\cdot z)= \underset{z',x'\sim U_l}{\E}[ext^l(C)((x,x'),(x\cdot z,x'\cdot z')]
\end{equation}
Consider
\begin{align*}\begin{split}
\underset{z'\sim U_{l}}{\E}[H_{ext^l(\mathcal{C}_v)}( z, z')]&\triangleq  \underset{z'\sim U_{l}}{\E} \underset{\substack{C\sim \mathcal{C}_v \\(x,x')\sim U_L}}{\E}[ext^l(C)( (x,x'), (x\cdot z,x'\cdot z'))] \quad\ldots\text{due to \cref{xorprotocol}} \\
&\triangleq  \underset{\substack{C\sim \mathcal{C}_v \\x\sim U_M}}{\E}[C( x, x\cdot z)] \quad\ldots\text{due to \cref{usefulrestriction}} \\
\end{split}\end{align*}
For each $j\in V$, let $a_j$ be a uniformly random sample as in step 2. For the rest of the coordinates $j\in [M]\setminus V$, set $a_j:=0$ and let $a=(a_1,\ldots,a_{M})\in \{-1,1,0\}^{M}$. Let $\mathcal{A}$ denote the distribution of $a$ obtained by this process. This, along with the above equation and the definition of $C_v$ implies that
\begin{align*}\begin{split}
\underset{z'\sim U_{l}}{\E}[H_{ext^l(\mathcal{C}_v)}( z, z')]&= \underset{\substack{C\sim \mathcal{C}_v \\x\sim U_M}}{\E}[C( x, x\cdot z)] =\underset{C\sim \mathcal{C}}{\E} \underset{\substack{a\sim \mathcal{A}\\ x\sim U_M}}{\E}[C(\rho_a(x), \rho_{a\cdot v}(x\cdot z))]
\end{split}\end{align*}
Note that $\rho_{a\cdot v}(x\cdot z)$ is exactly $\rho_a(x)\cdot \rho_{ v}(z).$ This is because for $j\in V$, we have $a_j,v_j\neq 0$ and thus, $(\rho_{a\cdot v}(x\cdot z))(j)=a_j \cdot v_j = (\rho_a(x))(j) \cdot (\rho_{v}(z))(j)$; similarly, for $j\notin V$, we have $a_j=v_j=0$ and thus, $(\rho_{a\cdot v}(x\cdot z))(j)=x(j)\cdot z(j)=(\rho_a(x))(j) \cdot (\rho_{v}(z))(j)$. Substituting this in the above equation,
\begin{align*}\begin{split}
\underset{z'\sim U_{l}}{\E}[H_{ext^l(\mathcal{C}_v)}( z, z')]&=\underset{C\sim \mathcal{C}}{\E} \underset{\substack{a\sim \mathcal{A}\\ x\sim U_M}}{\E}[C(\rho_a(x), \rho_{a}(x)\rho_v( z)) ]
\end{split}\end{align*}
Note that for $a\sim \mathcal{A}$ and $x\sim U_M$, we have $\rho_a(x)\sim U_M$. Substituting this in the above equation,
\begin{align*}\begin{split}
\underset{z'\sim U_{l}}{\E}[H_{ext^l(\mathcal{C}_v)}( z, z')]&=\underset{C\sim \mathcal{C}}{\E} \underset{ x\sim U_M}{\E}[C(x,x\cdot \rho_v( z)) ]\\
&=  \underset{z'\sim U_{l}}{\E}  \underset{\substack{C\sim \mathcal{C}\\ (x,x')\sim U_L}}{\E}[ext^l(C)((x,x'),(x\cdot \rho_v(z),x'\cdot z'))]\quad \ldots \text{due to \cref{usefulrestriction} }\\
&= \underset{ z'\sim U_{l} }{\E} [H_{ext^l(\mathcal{C})}(\rho_v(z),z')] \quad\ldots\text{due to \cref{xorprotocol}} \\
&=H(\rho_v(z)) \quad\quad\ldots\text{due to the definition in \cref{closedunderrestrictions}.}
\end{split}\end{align*}
This completes the proof of \cref{closedunderrestrictionslemma}.
\end{proof}
\section{Weight Bound}

For $l\in \mathbb{N}$, we say that a deterministic protocol $C:\{-1,1\}^M\times \{-1,1\}^M\rightarrow \{-1,1\}$ has minimum cost at least $l$, if every rectangle in the partition induced by the protocol has length and width at most $2^{M-l}.$

\begin{lemma} \label{weightbound1}  Let $C(x,y):\{-1,1\}^M\times \{-1,1\}^M \rightarrow \{0,1\}$ be any deterministic protocol of cost at most $c$ and of minimum cost at least $l:=\lceil 2k\log e\rceil$. Let $H:\{-1,1\}^M\rightarrow \mathbb{R}$ be defined at every $z\in \{-1,1\}^M$ by $H(z):=\underset{x\sim U_{M}}{\E}[C(x,x\cdot z)]$ as in \cref{xorprotocol}. Then,  $ L_{2k}(H)\le O\left( \left(\frac{e}{k}\right)^{2k}\cdot c^{2k} \right)$.
\end{lemma}

\begin{corollary} \label{weightbound2} Let $l=\lceil 2k\log e\rceil$. Let $\mathcal{C}$ be a distribution over deterministic protocols $C:\{-1,1\}^M\times \{-1,1\}^M\rightarrow \{-1,1\}$ of cost at most $c$. Let $H_{ext^l(\mathcal{C})}$ be as in \cref{xorprotocol}. Then,
$ L_{2k}(H_{ext^l(\mathcal{C})}) \le O\left( \left(\frac{e}{k}\right)^{2k}\cdot (c+2l)^{2k} \right)$.

\end{corollary}

\begin{proof}[Proof of \cref{weightbound3} using \cref{weightbound2}]
Let $H(z):=\E_{z'\sim U_l } [H_{ext^l(\mathcal{C})}(z,z')]$ be as in \cref{weightbound3}. Note that for all $S\subseteq[M]$, we have $\widehat{H}(S) =\widehat{H_{ext^l(\mathcal{C})}}(S)$. This implies that $L_{2k}(H)\le L_{2k}(H_{ext^l(\mathcal{C})})$. \cref{weightbound2} implies that $ L_{2k}(H_{ext^l(\mathcal{C})}) \le O\left( \left(\frac{e}{k}\right)^{2k}\cdot (c+2l)^{2k} \right)  $. This completes the proof of \cref{weightbound3}.
\end{proof}

\begin{proof}[Proof of \cref{weightbound2} using  \cref{weightbound1}]

Note that for all $S\subseteq[M+l]$, we have  $\widehat{H_{ext^l(\mathcal{C})}}(S) = \underset{C\sim \mathcal{C}}{\E} [\widehat{H_{ext^l(C)}}(S)]$. This, along with Triangle-Inequality implies that $L_{2k}(\mathcal{C})\le \max_{C\sim\mathcal{C}}L_{2k}(C)$. Let $C$ be any deterministic protocol in the support of $\mathcal{C}$. Note that $ext^l(C)$ is a deterministic protocol of cost at most $c+2l$ and of minimum cost $l$. Let $H_{ext^l(C)}$ be as in \cref{xorprotocol}. \cref{weightbound1} implies that
$ L_{2k}(H_{ext^l(C)}) \le O\left( \left(\frac{e}{k}\right)^{2k}\cdot (c+2l)^{2k} \right) $. This completes the proof of \cref{weightbound2}.
\end{proof}

\begin{proof}[Proof of \cref{weightbound1} ]
In order to bound $L_{2k}(H)$, we will use the following lemma. Its statement and proof appear as `Level-$k$ Inequalities' on Page 259 of `Analysis of Boolean Functions' \cite{odonnell}. For $S\subseteq\{-1,1\}^n$, let $\mathbbm{1}_S:\{-1,1\}^n\rightarrow \{0,1\}$ denote the $\{0,1\}$-indicator function of the set $S$, that is, for $x\in \{-1,1\}^n$, let $\mathbbm{1}_S(x)=1$ if and only if $x\in S$.
\begin{lemma}[Level-$k$ Inequalities] \label{levelkinequality}  Let $A\subseteq\{-1,1\}^n$ be a set such that $\E[\mathbbm{1}_A]=\alpha$ and let $k\in \mathbb{N}$ be at most $2\ln(1/\alpha)$. Then,
\[ \sum_{|S|=k} \left(\widehat{\mathbbm{1}_A}(S)\right)^2\le\alpha^2 \left(\frac{2e}{k}\ln(1/\alpha)\right)^k \]
\end{lemma}

We now show the desired bound on $L_{2k}(H)$. Since $C$ is a deterministic protocol of cost at most $c$, it induces a partition of the input space $\{-1,1\}^{M}\times \{-1,1\}^{M}$ into at most $2^{c}$ rectangles. Let $\mathcal{P}$ denote the set of rectangles in this partition and let $A\times B$ index these rectangles, where $A$ (respectively $B$) is the set of Alice's (respectively Bob's) inputs compatible with the rectangle. Let $C(A\times B)\in \{-1,1\}$ denote the output of the protocol when the inputs are in $A\times B$. For all $x,y\in \{-1, 1\}^M$,
\[ C(x,y) = \underset{A\times B \in \mathcal{P}}{\sum} C(A\times B) \mathbbm{1}_{A}(x)  \mathbbm{1}_B(y)\]
This implies that for all $x,z\in \{-1, 1\}^M,$
\[ C(x,x\cdot z) = \underset{A\times B \in \mathcal{P}}{\sum} C(A\times B)\mathbbm{1}_{A}(x)  \mathbbm{1}_B(x\cdot z)\]
Taking an expectation over $x\sim U_M$ of the above identity implies that
\[ H(z)\triangleq \underset{x\sim U_{M}}{\E} [C(x,x\cdot z)]= \sum_{A\times B\in \mathcal{P}}C(A\times B)\big(\mathbbm{1}_A*\mathbbm{1}_B\big)(z) \]
This implies that for any $S\subseteq [M]$,
\[\widehat{H}(S)= \sum_{A\times B\in \mathcal{P}} C(A\times B)\widehat{\mathbbm{1}_A*\mathbbm{1}_B}(S)=\sum_{A\times B\in \mathcal{P}} C(A\times B)\widehat{\mathbbm{1}_A}(S)\widehat{\mathbbm{1}_B}(S)\]
Note that $C(A\times B)\in \{-1,1\}$. We thus obtain
\begin{align*}\begin{split}
L_{2k}(H)&=\sum_{|S|=2k}\left| \widehat{H}(S) \right| \\
&= \sum_{|S|=2k}\left| \underset{A\times B\in \mathcal{P}}{\sum} C(A\times B) \widehat{\mathbbm{1}_A}(S)\widehat{\mathbbm{1}_B}(S)\right| \\
&\le \underset{A\times B\in \mathcal{P}}{\sum} \sum_{|S|=2k}  |\widehat{\mathbbm{1}_A}(S)| |\widehat{\mathbbm{1}_B}(S)|\\
\end{split}\end{align*}
We apply Cauchy Schwarz to the term $\sum_{|S|=2k}  |\widehat{\mathbbm{1}_A}(S)| |\widehat{\mathbbm{1}_B}(S)|$ to obtain
\[ L_{2k}(H) \le  \underset{A\times B\in \mathcal{P}}{\sum} \Big( \sum_{|S|=2k} \widehat{\mathbbm{1}_A}(S)^2 \Big)^{1/2}\Big( \sum_{|S|=2k} \widehat{\mathbbm{1}_B}(S)^2 \Big)^{1/2} \]
For ease of notation, let $\mu(A)=\frac{|A|}{2^M}$ denote the measure of a set $A\subseteq \{-1, 1\}^M$ under $U_M$. Because of the assumption that the minimum cost of $C$ is at least $l=\lceil 2k\log e\rceil$, every rectangle $A\times B\in \mathcal{P}$ satisfies $\mu(A) , \mu(B) \le e^{-2k}$. This ensures that $2k\le 2\ln\frac{1}{\mu(A)}$ and $2k\le 2\ln\frac{1}{\mu(B)}$. We apply  \cref{levelkinequality} on the indicator functions $\mathbbm{1}_A$ and $\mathbbm{1}_B$ at level $2k$ to obtain
\[\sum_{|S|=2k} \left(\widehat{\mathbbm{1}_A}(S)\right)^2 \le \mu(A)^2\Big(\frac{2e}{2k}\cdot \ln(1/\mu(A))\Big)^{2k} \]
\[ \sum_{|S|=2k} \left(\widehat{\mathbbm{1}_B}(S)\right)^2 \le \mu(B)^2\Big(\frac{2e}{2k} \cdot \ln(1/\mu(B))\Big)^{2k} \]
Substituting this in the bound for $L_{2k}(H)$, we have
\[ L_{2k}(H) \le \left(\frac{e}{k}\right)^{2k} \underset{A\times B\in \mathcal{P}}{\sum} \mu(A)\mu(B)\left( \ln\frac{1}{\mu(A)}\ln\frac{1}{\mu(B)}\right)^{k} \]
Let $\Delta:= \left(\frac{e}{k}\right)^{2k}\underset{A\times B\in \mathcal{P}}{\sum} \mu(A)\mu(B)\left(\ln\frac{1}{\mu(A)}\ln\frac{1}{\mu(B)}\right)^k$ be the expression in the R.H.S. of the above. Consider the case when $\mathcal{P}$ consists of $2^c$ rectangles $A\times B$, each of which satisfies $\mu(A)=\mu(B)=\frac{1}{2^{c/2}}$. In this case, $\Delta$ evaluates to $ \left(\frac{e}{k}\right)^{2k} \sum_{A\times B\in \mathcal{P}} \frac{1}{2^c}( \frac{c\ln2}{2})^{2k} =O\left( \left(\frac{e}{k}\right)^{2k} \cdot c^{2k}\right)$. This proves the lemma in this special case. A similar bound holds for the general case and the proof follows from a concavity argument that we describe now.

\begin{figure}
\centering
\begin{tikzpicture}[scale=0.7]
    \begin{axis}[
            axis lines=middle,
            xmin=0,xmax=0.3,ymin=0,ymax=6,
            xlabel=$x$,
            ylabel=$y$,
    y label style={at={(axis description cs:0.05,1.0)}},
    x label style={at={(axis description cs:1.0,0.05)}},
            ]
      \addplot[domain=0:0.3,thick,samples=100] {x*ln(1/x)^4};
    \end{axis}
  \end{tikzpicture}
\caption{Plot of the function $y=x\left(\ln\frac{1}{x}\right)^4$}
\end{figure}
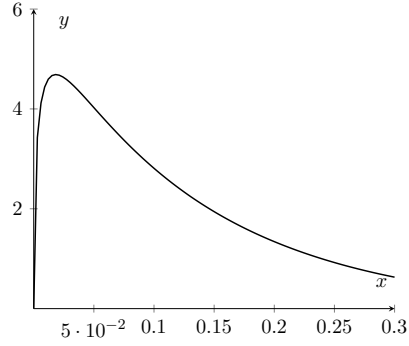

Since $\mu(A),\mu(B)\le 1$, we have the following inequality.
\begin{align*}\begin{split}
\Delta &\triangleq  \left(\frac{e}{k}\right)^{2k} \underset{A\times B\in \mathcal{P}}{\sum} \mu(A)\mu(B)\left(\ln\frac{1}{\mu(A)}\ln\frac{1}{\mu(B)}\right)^k\\
&\le  \left(\frac{e}{k}\right)^{2k}  \underset{A\times B\in \mathcal{P}}{\sum} \mu(A)\mu(B)\left(\ln\frac{1}{\mu(A)\mu(B)}\ln\frac{1}{\mu(A)\mu(B)}\right)^k\\
&= \left(\frac{e}{k}\right)^{2k}  \underset{A\times B\in \mathcal{P}}{\sum} \mu(A\times B)\left(\ln\frac{1}{\mu(A\times B)}\right)^{2k}
\end{split}\end{align*}

Let $f:[0,\infty)\rightarrow \mathbb{R}$ be defined by $f(p):=p\ln (1/p)^{2k}$. A small calculation\footnote{Consider $f'(p)=\ln(1/p)^{2k} -2k \ln(1/p)^{2k-1}$. This implies that $f''(p)=2k\ln(1/p)^{2k-2}\cdot\frac{1}{p}\cdot\left( (2k-1) - \ln(1/p) \right)$. Note that for $p\le \frac{1}{e^{2k-1}}$, $f''(p)\le 0$.} shows that $f$ is a concave function in the interval $[0,\frac{1}{e^{2k-1}}]$ (see Figure 2). Let $\alpha_i\in [0,\frac{1}{e^{2k-1}}]$ for $i\in [d]$. Jensen's inequality applied to $f$ states that for $i\sim [d]$ drawn uniformly at random, we have $\E_i [f(\alpha_i)]\le f(\E_i [\alpha_i])$. This implies that
\[ \sum_{i=1}^d \alpha_i  \ln(1/\alpha_i)^{2k} \le \left(\sum_{i=1}^d\alpha_i\right)\ln\left(\frac{d}{\sum_{i=1}^d \alpha_i}\right)^{2k} \]
We apply this inequality to the terms in $\Delta$ by substituting $\alpha_i$ with $\mu(A\times B)$. We may do this because of the assumption that $\mu(A),\mu(B)\le \frac{1}{e^{2k}}$. This implies that
\[ \Delta \le \left(\frac{e}{k}\right)^{2k} \left( \sum_{A\times B\in \mathcal{P}}\mu(A\times B)\right) \ln\left( \frac{2^{c}}{\sum_{A\times B\in \mathcal{P}} \mu(A\times B)} \right)^{2k}  \]
Note that $\sum_{A\times B\in \mathcal{P}}\mu(A\times B)= 1$. This, along with the above inequality implies that $\Delta \le O\left( \left(\frac{e}{k}\right)^{2k}\cdot c^{2k} \right)$. This completes the proof of \cref{weightbound1}.\end{proof}

\end{document}